\documentclass[11pt,a4paper]{article}
\usepackage{natbib}
\usepackage{amsfonts}
\usepackage{amsthm}
\usepackage{graphicx}
\usepackage[ansinew]{inputenc}

\newtheorem{lemma}{Lemma}
\newtheorem{theorem}{Theorem}

\title{Iterative algorithms for total variation-like reconstructions in seismic tomography}
\author{Ignace Loris and Caroline Verhoeven\\Université Libre de Bruxelles, Brussels, Belgium.}

\begin{document}
\maketitle
\begin{abstract}
A qualitative comparison of total variation like penalties (total variation, Huber variant of total variation, total generalized variation, \ldots) is made in the context of global seismic tomography. Both penalized and constrained formulations of seismic recovery problems are treated. A number of simple iterative recovery algorithms applicable to these problems are described. The convergence speed of these algorithms is compared numerically in this setting. For the constrained formulation a new algorithm is proposed and its convergence is proven.
\end{abstract}

\section{Introduction}

The aim of this paper is two-fold. To give a qualitative description of various total variation-like regularization methods in the context of a global seismic inversion problem, and to compare several iterative algorithms numerically, that can be used to perform these inversions. A new algorithm is also included and a proof of convergence is given.

Inverse problems in seismic tomography are characterized by a combination of insufficient and noisy data. It follows that resulting tomographic reconstructions of wave-speed anomalies are non-unique and corrupted by noise; some kind of regularization of the inverse problem is required.
The dynamics of the Earth's mantle leads to smooth variations (as a result of heat diffusion) as well as steep gradients (related to the exponential dependence of viscosity on temperature), and sharp transitions (resulting from chemical and mineralogical variations caused by lithosphere subduction and from phase transitions). An important challenge is to retain as much of this in the reconstruction despite the need to regularize. The attraction of simple smoothing algorithms giving simple images, is off-set by smoothing away sharp boundaries.

Here we focus on regularization techniques that have the ability to reconstruct sharp edges. Included in our discussion are the total variation prior (TV) of \cite{Rudin.Osher.ea1992} and a number of generalizations and variations on this theme. Among them are the so-called Huber variant of the total variation prior (HTV), the total generalized variation method (TGV) of \cite{Bredies.Kunisch.ea2010}, etc. We also include a regularization method that uses a sparse expansion in terms of wavelet coefficients. The methods are described in explicit detail and some typical features of resulting reconstructions are discussed qualitatively in Section~\ref{reviewsection}.

Although details differ, the common theme of these penalties is that they all impose sparsity of certain local differences of the reconstructed model, which is achieved in practice by using a non-smooth convex $\ell_1$-norm penalty of these local differences. The effect of using an $\ell_1$-norm instead of an $\ell_2$-norm squared is that in practice small coefficients are penalized disproportionately more than large coefficients, which leads to sparse reconstructions (i.e. reconstructions with few nonzero coefficients).

The use of an $\ell_1$-norm penalty implies that non-linear equations have to be solved in the inversion. We therefore devote a section of the paper to the numerical comparison of the convergence speed of some iterative reconstruction algorithms for these problems. Indeed, a second common feature of these regularization methods is that they can be solved with very similar algorithms (with minor changes).

In this paper we will assume a linear relationship between an unknown model $u$ and data $y$ characterized by a matrix $K$. One may think e.g. of the matrix $K$ as containing in its rows the discretized versions (on some grid) of ray or finite frequency sensitivity kernels. The central mathematical theme of this paper is then the numerical minimization of some penalized least squares functionals of type:
\begin{equation}
\min_u \frac{1}{2}\|Ku-y\|^2+\mathrm{penalty}.\label{penproblem1}
\end{equation}
The first term in this functional represents a quadratic data misfit term that depends on the data $y$ and the matrix $K$: $\|Kx-y\|^2=\sum_i(Kx-y)_i^2$. The second term depends on $u$ and serves to regularize the inversion, i.e. it serves to produce a unique model $u$ that satisfies some (qualitative) assumptions imposed on the reconstruction. Here we are principally interested in penalizing local differences $Au$ of the model $u$ in such a way that edges in the model aren't blurred too much. To this end, we will study the problem:
\begin{equation}
\hat u(\lambda)=\arg\min_u \frac{1}{2}\|Ku-y\|^2+\lambda \|Au\|_1.\label{penproblem}
\end{equation}
that contains a non-smooth convex $\ell_1$-norm penalty of local differences $Au$ of $u$.
The precise choice of the differencing matrix $A$ and the precise form of the $\ell_1$-norm $\|\cdot\|_1$ depend on the penalty that is preferred and on the particularities of the model (2D, 3D etc,...). Fully worked out examples of this kind of penalty, which include the TV penalty and generalizations, are presented in Section~\ref{reviewsection}, with further details for solving the associated problem (\ref{penproblem}) in Sections~\ref{algsection} and \ref{detailssection}.

Alternatively, instead of penalizing a least squares term, one may try to solve the constrained problem:
\begin{equation}
\tilde u(\epsilon)=\arg\min_{\|Ku-y\|\leq\epsilon} \|Au\|_1.\label{conproblem}
\end{equation}
The problems (\ref{penproblem}) and (\ref{conproblem}) are equivalent in the sense that for corresponding $\lambda$ and $\epsilon$, the models $\hat u(\lambda)$ and $\tilde u(\epsilon)$ are equal; see e.g. \citep{Hennenfent.Berg.ea2008,Berg2011} for the case $A=1$. The formulation (\ref{conproblem}) is useful as the parameter $\epsilon$ (related in practice to the amount of noise on the data $y$) is perhaps easier to estimate than the penalty parameter $\lambda$ in (\ref{penproblem}).

In Section~\ref{algsection} we write explicit formulas for several iterative algorithms for the penalized problem (\ref{penproblem}) and we compare the speed of convergence (in the special case of the TV penalty). The comparison is done for a matrix $K$ that has no special structure (not the identity matrix, not convolution, ...). We also discuss two iterative algorithms for the constrained problem (\ref{conproblem}). One we believe is new, and we prove convergence. We also make a numerical comparison of the speed of convergence in the constrained case.

Except for one, the algorithms discussed in this paper are fully explicit. They require only application of matrix vector products (with matrices $A$ and $K$), and one or two simple convex projections. It is important to remark that only the precise form of these projections and the choice of the local differencing operator $A$ differ between the various penalization methods discussed in this paper. Again fully worked out examples are given for the penalties discussed in Section~\ref{reviewsection}.

In section~\ref{reviewsection} we will review a number of non smooth regularization terms in the context of a toy synthetic seismic tomography experiment. We discuss the effect the various penalty choices have on the reconstruction. Section~\ref{algsection} lists a number of iterative algorithms that can be used to perform these inversions. A numerical convergence speed comparison is made, both in the penalized case (\ref{penproblem}) as well as in the constrained case (\ref{conproblem}).
Section~\ref{detailssection} explains in detail how these algorithms can be used to solve the problem introduced in Section~\ref{reviewsection}, i.e. we give explicit expressions for the convex projection operators used by the various penalties. Finally, Section~\ref{proofsection} contains a proof of convergence of a new iterative minimization algorithm for the constrained problem (\ref{conproblem}).

\section{Comparison of edge-preserving regularization methods}
\label{reviewsection}

Natural images are often characterized by the presence of both sharp edges and smooth transitions. It is therefore not surprising that much research is done for finding methods of denoising or deconvolving images that can preserve sharp edges. A popular technique is the use of the total variation prior \citep{Rudin.Osher.ea1992}, and many generalizations have been proposed. In this section we review a number of these TV-like techniques, compare synthetic 2D reconstructions of a toy model in global seismic tomography and describe their principal features. The description of computational algorithms for these tasks is deferred to Section~\ref{algsection} and specific implementation details for each of the cases discussed are given in Section~\ref{detailssection}.

\subsection{Total variation penalty and generalizations}
\label{overviewsubsection}

The total variation (TV) penalty is defined by the $\ell_1$-norm of the gradient of the model $u$, i.e. we choose $A=\mathrm{grad}$ so that the penalty in expressions (\ref{penproblem1}) and (\ref{penproblem}) becomes:
\begin{equation}
\mathrm{penalty}=\lambda\|Au\|_1=\lambda\|\mathrm{grad}(u)\|_1=\lambda\sum_\mathrm{pixels}\sqrt{(\Delta_x u)^2+(\Delta_y u)^2}.
\label{TVdef}
\end{equation}
Here $\mathrm{grad}(u)=(\Delta_xu,\Delta_yu)$ and $\Delta_xu$ and $\Delta_yu$ are first order local differences of the 2D model $u$. In formula (\ref{TVdef}) the sum ranges over all pixels of the model $u$. This type of regularization will force the gradient of the model to be (exactly) zero in many places: TV promotes sparsity of the gradient of $u$. In other words, it will give rise to a piece-wise constant reconstruction whenever the data allows for it. This property is lost when the square root in expression (\ref{TVdef}) is removed. This penalty is translation invariant and isotropic (invariant under rotations). Sometimes a non-isotropic TV is used which is defined as a sum over all pixels of $|\Delta_x u|+|\Delta_y u|$. Formula (\ref{TVdef}) is easily adapted to 3D models too.

A slight generalization of the TV penalty is the Huber total variation penalty (HTV). Here we replace the $\sqrt{(\Delta_x u)^2+(\Delta_y u)^2}$ terms of (\ref{TVdef}) by:
\begin{equation}
\mathrm{penalty}=\lambda\sum_\mathrm{pixels}h\left(\sqrt{(\Delta_x u)^2+(\Delta_y u)^2}\right)
\label{HTVdef}
\end{equation}
where the function $h$ is defined by
\begin{equation}
h(t)=\left\{
\begin{array}{lcl}
|t|^2/2\alpha &\quad\mathrm{if}\quad& |t|\leq\alpha\\
|t|-\alpha/2 &\mathrm{if}& |t|\geq\alpha
\end{array}\right.
\end{equation}
(see \cite[Section~4, point (iii)]{Huber1964}).
Clearly HTV reduces to TV for $\alpha=0$. On the other hand, when $\alpha$ is positive, large gradients are penalized as in TV (up to an irrelevant constant), but small gradients are penalized quadratically so that the sparse gradient promoting property of the TV penalty is lost. In practice, local differences will be kept small by this penalty (if possible), but non will be exactly zero as in TV.

A third kind of regularization method that we will consider is a (non-symmetric) total generalized variation penalty (TGV) of \cite[remark 3.10]{Bredies.Kunisch.ea2010}.
It also involves an auxiliary variable $v=(v_x,v_y)$ and the problem (\ref{penproblem}) takes the form
\begin{equation}
\begin{array}{lcl}
\displaystyle\min_{x,v}\frac{1}{2}\|Ku-y\|^2\\[2mm]
\displaystyle+\lambda
\sum_\mathrm{pixels}\sqrt{(\Delta_x u-v_x)^2+(\Delta_y u-v_y)^2}+\alpha
\sqrt{(\Delta_x v_x)^2+(\Delta_y v_x)^2+(\Delta_x v_y)^2+(\Delta_y v_y)^2}
\end{array}\label{GTVdef}
\end{equation}
(see Section~\ref{detailssection} for more details).
This generalization was designed to yield piece-wise smooth models (instead of piece-wise constant models as in standard TV). The penalty depends on an additional parameter $\alpha$ that controls the balance of the first and second term. Many other generalizations of the TV penalty exist; see e.g. \citep{Chambolle167--188,Bredies.Kunisch.ea2010,Chan2000}.

A straightforward variation on the TV penalty is the use of the $\ell_1$-norm, not of the gradient of the model, but of second order differences of the model (i.e. of the local Hessian matrix; $A=\mathrm{Hess}$). We will call it the Hessian penalty (HP):
\begin{equation}
\begin{array}{lcl}
\mathrm{penalty}&=&\displaystyle\lambda\|Au\|_1=\lambda\|\mathrm{Hess}(u)\|_1\\
&=& \displaystyle\lambda\sum_\mathrm{pixels}\sqrt{
(\Delta_x^2u)^2+(\Delta_x\Delta_y u)^2+
(\Delta_y\Delta_xu)^2+(\Delta_y^2 u)^2
}
\end{array}
\label{Hessdef}
\end{equation}
(in each pixel $\mathrm{Hess}(u)$ is a $2\times2$ matrix; see again Section~\ref{detailssection} for details). The use of penalty (\ref{Hessdef}) will force the local Hessian of the model to be zero in most places, i.e. the model will be piecewise linear.
Here, any matrix norm on the local Hessian can be chosen; those matrix norms that are expressed in terms of the eigenvalues will yield a penalty that is isotropic. We choose the Frobenius norm as it is easy to work with and isotropic (an explicit expression is given in Section~\ref{detailssection}).

Regularization strategies that try to express the model as a sparse linear combination of wavelet basis functions \citep{Mallat2009} also fit in the category of penalties that use the $\ell_1$-norm of local differences. In that case one has
\begin{equation}
\mathrm{penalty}=\lambda\|Au\|_1=\lambda\|Wu\|_1\label{waveletpenalty}
\end{equation}
were $W$ represents a wavelet transform. A wavelet transform $W$ interleaves local differencing (and averaging) with subsampling. Such a penalty depends directly on the choice of wavelet basis (choice of $W$). On a cartesian grid, many wavelet families exist \citep{Mallat2009} and they are relatively easy to implement (although not as easy as the local differencing used in TV and its variations). Examples include non-smooth wavelet functions such as the simple Haar wavelets, smooth orthogonal wavelets or smooth symmetric wavelets. However, wavelets are non-stationary (preferred positions exist) and 2D and 3D wavelets are non-isotropic (preferred directions exist). Partial solutions to these two problems include the use of the undecimated WT or of directional transforms such as curvelets \citep{Candes2006} or shearlets \citep{Labate.Lim.ea2005}. Such a strategy has been used in geosciences in \citep{Loris.Nolet.ea2007,Hennenfent.Berg.ea2008,herrmann08crsi,Gholami.Siahkoohi2010,Simons.Loris.ea2011} etc.

\subsection{Qualitative comparison on a synthetic example}

In the remaining part of this section we apply the various regularization methods described above to a toy problem in seismic tomography. We consider a simple synthetic $2D$ input model $u^\mathrm{input}$ defined on the globe (see Figure~\ref{inputfigure}, panels a--e).
The input model is chosen to have a number of zones of constant value ($+1$ in blue color or
$-1$ in red color) with either a sharp edge or a smooth transition in between. Sharp edges are found near North-America. Smooth transitions are found around the Indian Ocean. The edges and transitions are circle-shaped so as not to give preference to any particular
direction (i.e. edges are not aligned with the parametrization grid of the sphere). The center of this `bull's eye' pattern is located at $(36^{\circ} \mathrm{N}, -120^{\circ} \mathrm{E})$. The model has circular symmetry around this point.

In panel (b) of Figure~\ref{inputfigure}, a cross section of this model is shown along a great circle passing through the point $(36^{\circ} \mathrm{N}, -120^{\circ} \mathrm{E})$ and through the North Pole. The second row of Figure~\ref{inputfigure} (panels c and d) depict the length of the local gradient of the input model: $|\mathrm{grad}(u^\mathrm{input})|$. The gradient is sparse; nonzero values can only be found near edges and transitions in the model. The `cubed sphere'  \citep{Ronchi.Iacono.ea1996} was used as a parametrization of the sphere. In this toy experiment the model space has dimension $98304$ ($=128\times128\times 6$ pixels).

In order to set up a synthetic inverse problem, we use a set of $8490$ seismic rays  corresponding to actual earthquake positions and seismic stations \citep{Trampert+95,Trampert+96a,Trampert+2001}. The rays (discretized on the $128\times128\times6$ grid) make up the rows of a $8490\times 98304$ matrix $K$. In panel (f) of Figure~\ref{inputfigure} the sum of all rows of the matrix $K$ is shown. It represents the illumination of the globe by the $8490$ ray paths. Most rays are concentrated around the Pacific Ocean. The matrix $K$ does not have any special structure that can be exploited by a minimization algorithm.

\begin{figure}
\centering\resizebox{\textwidth}{!}{\includegraphics{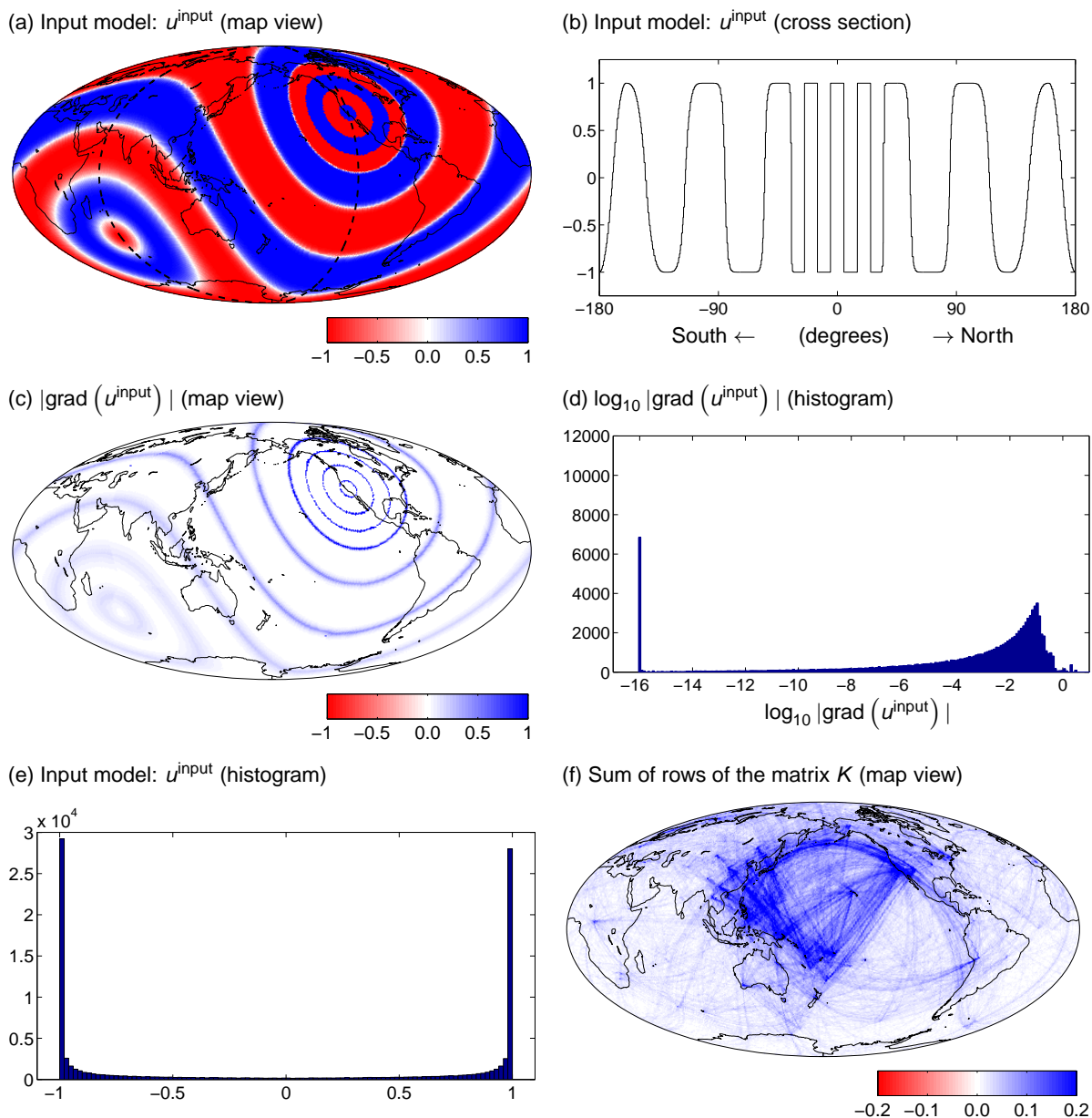}}
\caption{Toy input model for synthetic seismic tomography experiment (see Section~\ref{reviewsection}): (a) input model $u^\mathrm{input}$ with both sharp and smooth edges between zones of constant model value; (b) cross section along a great circle passing through $(36^{\circ} \mathrm{N}, -120^{\circ} \mathrm{E})$ (dashed line in panel (a)); the horizontal axis measures the degrees of separation from this point in Northern direction; (c) length of the local gradient of input model (mostly zero except at edges and transitions); (d) histogram of the length of the gradient (on a logarithmic scale; the peak at $-16$ corresponds zero gradients); (e) Histogram of model values; (f) Illumination of the globe by the $8490$ ray paths in the matrix $K$.}\label{inputfigure}
\end{figure}

Next, artificial data $y$ are constructed using the formula $y=Ku^\mathrm{input}+n$, where $n$ is chosen as gaussian noise of magnitude $\|n\|=0.1\times\|Ku^\mathrm{input}\|$ (in other words $10\%$ gaussian noise is added). Our goal is to use the different regularization techniques described in Section~\ref{overviewsubsection} to obtain faithful reconstructions of $u^\mathrm{input}$ (from the knowledge of $y$ and $K$ only), and to compare some of their principal characteristics.

We do not expect perfect reconstruction ($u^\mathrm{output}\neq u^\mathrm{input}$) because the problem is too under-determined (only $8490$ data for $98304$ unknowns), and because the data contain noise. The use of TV style regularization methods is appropriate as the model has several areas of constant model value, together with some edges. However, we expect that the TV penalty will unfortunately also enforce piecewise constant model values near the smooth transitions. We compare with the other regularization methods.

\begin{figure}
\centering\resizebox{\textwidth}{!}{\includegraphics{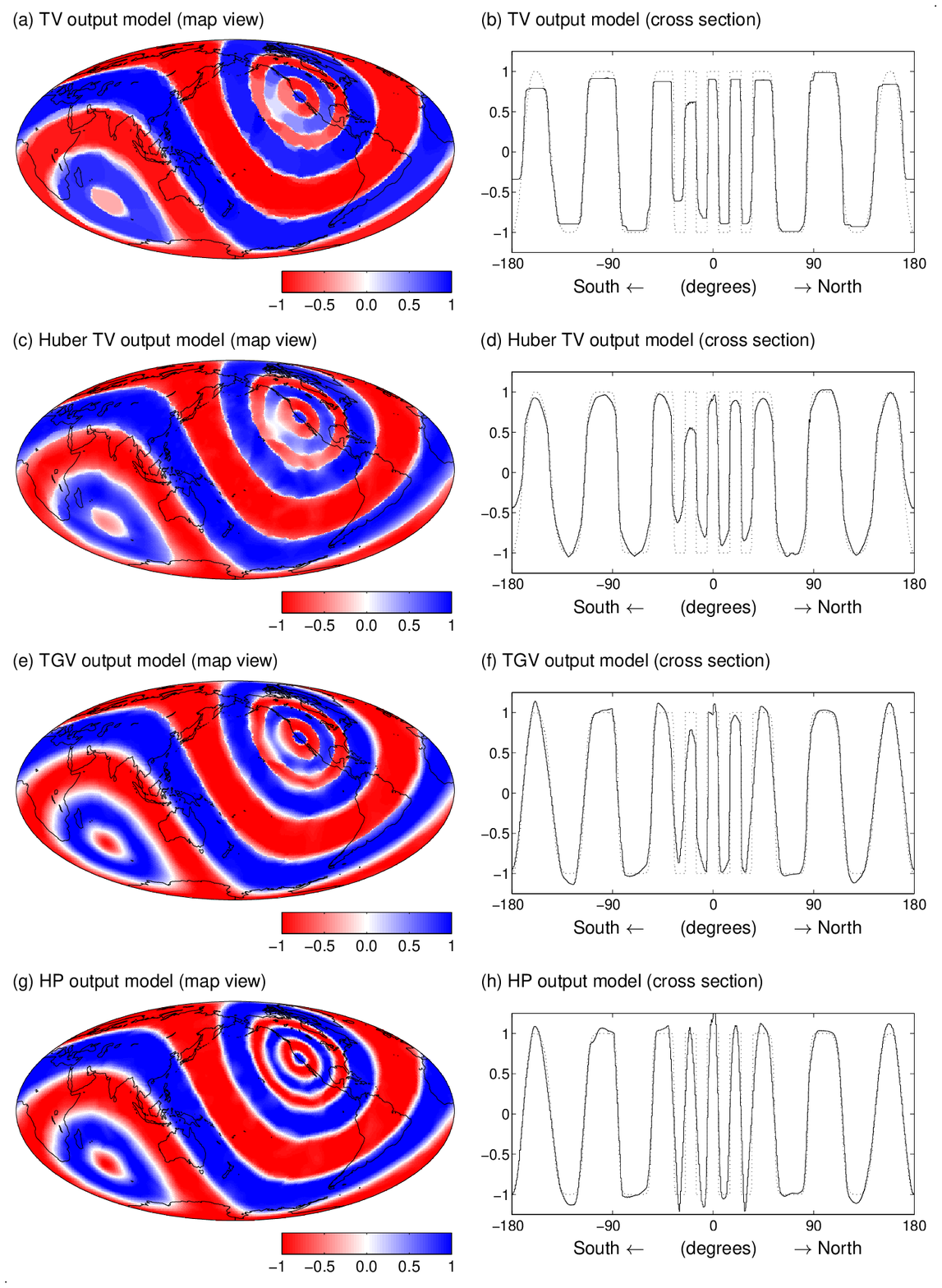}}
\caption{Reconstructions of the input model for four different regularization methods. The cross sections on the right show the output models (solid lines) and the input model (dashed lines). (a)-(b) the total variation reconstruction has an obvious piecewise constant character. The amplitude of the model appears damped in some regions. (c)-(d) the Huber-TV reconstruction is similar to the TV reconstructions but the staircasing effect is less pronounced. (e)-(f) the TGV reconstruction. (g)-(h) the Hessian style reconstruction is piecewise linear.}\label{reconstructionfigure}
\end{figure}

\begin{figure}
\centering\resizebox{\textwidth}{!}{\includegraphics{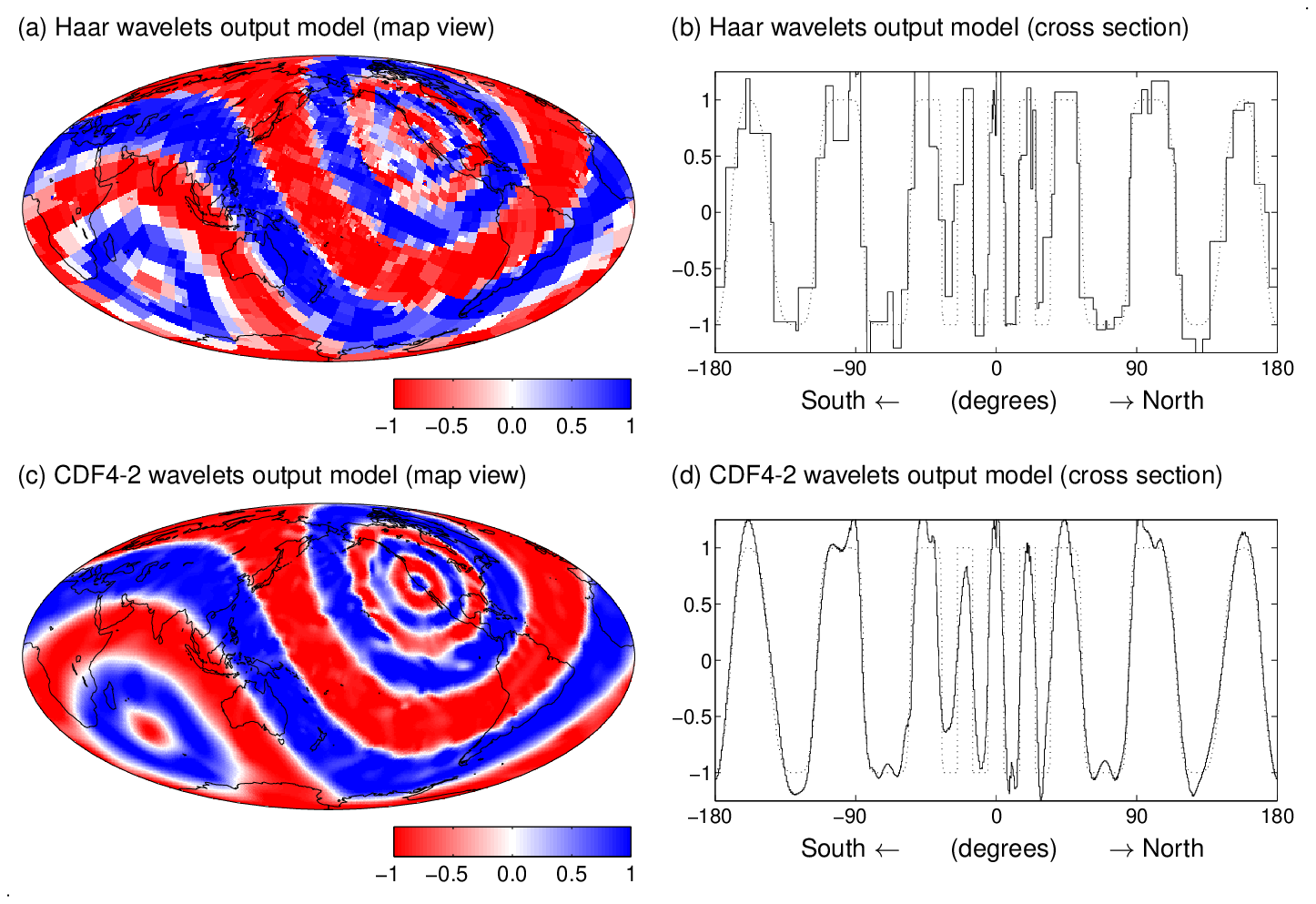}}
\caption{Reconstructions of the input model for two different wavelet regularization methods. The cross sections on the right show the output models (solid lines) and the input model (dashed lines). (a)-(b) a reconstruction that is sparse in the Haar wavelet basis. (c)-(d) a reconstruction that is sparse in the CDF4-2 wavelet basis. The Haar model uses only $1668$ nonzero wavelet basis coefficients, and the CDF 4-2 reconstruction uses a mere $996$ nonzero wavelet basis coefficients.}\label{waveletreconstructionfigure}
\end{figure}

Four different reconstructions are made, corresponding to TV, HTV, TGV and HP penalties as described by formulas (\ref{TVdef}), (\ref{HTVdef}), (\ref{GTVdef}) and (\ref{Hessdef}) in combination with functional (\ref{penproblem}). The precise iterative algorithm that is used to solve these four instances of problem (\ref{penproblem}) is described in Section~\ref{algsection}, formula (\ref{gista}). A detailed description is deferred to Section~\ref{detailssection}. At the moment we limit ourselves to discussing some qualitative differences of the resulting output models. It is important to mention right away that all reconstructed models fit the data equally well: $\|Ku^\mathrm{output}-y\|=\|n\|$.
The four reconstructions are displayed in Figure~\ref{reconstructionfigure}.

In Figure~\ref{reconstructionfigure}, a map view of the four reconstructed models is displayed on the  left hand side and a cross section of the four output models is shown in the right hand side column (for reference also a cross section of the input model is shown in dotted lines). As before the cross sections follow the great circle that passes through the point $(36^{\circ} \mathrm{N}, -120^{\circ} \mathrm{E})$ and through the North pole. Note however that the output models have lost (some of) their circular symmetry, and a different cross section (along another great circle through the point $(36^{\circ} \mathrm{N}, -120^{\circ} \mathrm{E})$) will yield a slightly different profile.

The piece-wise constant nature of the TV output model is clearly noticeable in panels (a) and (b). There is also some loss of amplitude in certain regions. Distinctively, sharp edges (e.g. near North America) are preserved but smooth transitions of the input model (e.g. near Africa and the Indian Ocean) are replaced by a succession of sharp edges (the so-called staircasing effect). In other words a TV penalty imposes a piecewise constant reconstruction as much as the data allows for it. Such a reconstruction always has sharp edges, but it is not guaranteed that the reconstructed edge will be at exactly the same location as in the input model (the position is influenced by noise on the data and by the limited number of rays).

The Huber TV reconstruction looks similar to the TV reconstruction but it has edges that are somewhat smoother. The piecewise constant nature of the reconstruction is lost. It should be said that the precise amount of smoothing depends on the value of the extra parameter $\alpha$ that is present in the Huber penalty. For $\alpha=0$ one recovers the TV reconstruction, but for large $\alpha$ one obtains a reconstruction that is penalized by the $\ell_2$-norm squared of the gradient.

The TGV reconstruction exhibits edges that are comparable to the Huber TV reconstructions, but amplitudes are damped less. Some edges are replaced by linear transitions. The TGV penalty also depends on a parameter $\alpha$ that can be tuned between the TV case (large $\alpha$) and the HP case (small $\alpha$). Linear transitions also characterize the HP reconstruction. Here the penalty is proportional to the $1$-norm of the second derivative of the model. For the example shown, the difference between the TGV reconstruction lies somewhere between the TV and HP reconstruction.

Finally, Figure~\ref{waveletreconstructionfigure} shows two reconstructions that were obtained by imposing a $\ell_1$-norm penalty of wavelet coefficients of the model. Panels (a) and (b) of this Figure show a reconstruction that is sparse in the Haar wavelet basis. The number of nonzero basis coefficients is $1668$ (out of a possible $98304$). These wavelets are orthonormal, but non-smooth. The second example (panels c and d) uses a wavelet family with very smooth wavelets (CDF 4-2) of \cite{Cohen1992}. This model has only $996$ nonzero coefficients (out of a possible $98304$). The smoothness of the wavelet bases functions is clearly reflected in the nature of the reconstructions shown in Figure~\ref{waveletreconstructionfigure}. These two reconstruction fit the $8490$ data as well as the four previous reconstructions. The Haar wavelet reconstruction is not visually appealing. Even though the Haar basis functions are piece-wise constant, the Haar reconstruction does not look similar to the TV reconstruction. If wavelets are used, one not only needs to choose the wavelet family (Haar, \ldots), but also the number of levels of the wavelet transform. In this example, we chose to make a 4 level wavelet transform.

\begin{figure}
\centering\resizebox{\textwidth}{!}{\includegraphics{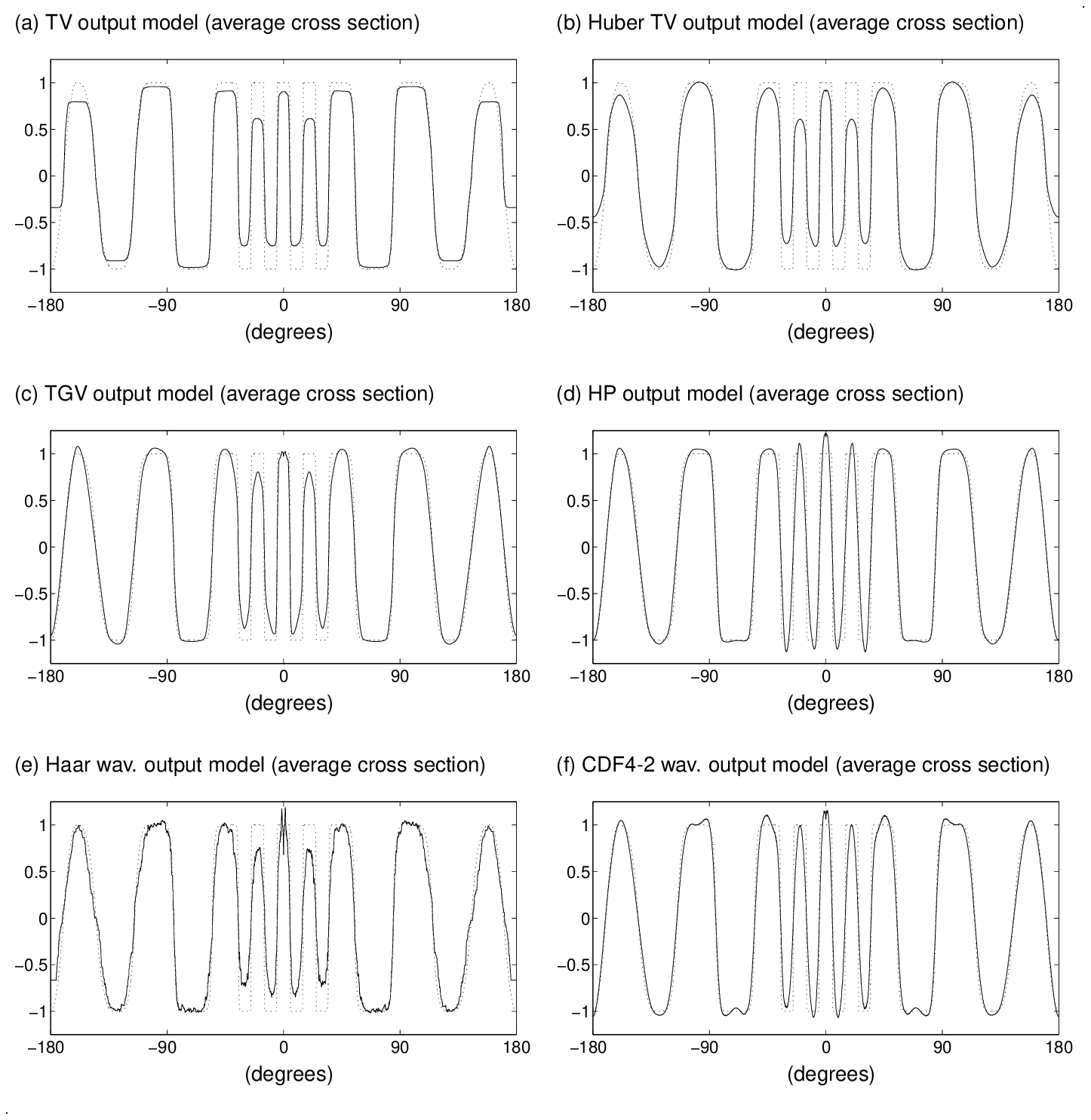}}
\caption{Averaged cross sections of the six reconstructions. The average is taken over $36$ great circles passing through the point $(36^{\circ} \mathrm{N}, -120^{\circ} \mathrm{E})$, and with azimuth $0,10,20,\ldots,350$ degrees w.r.t local North.}\label{averagecrosssectionfigure}
\end{figure}

\begin{figure}
\centering\resizebox{\textwidth}{!}{\includegraphics{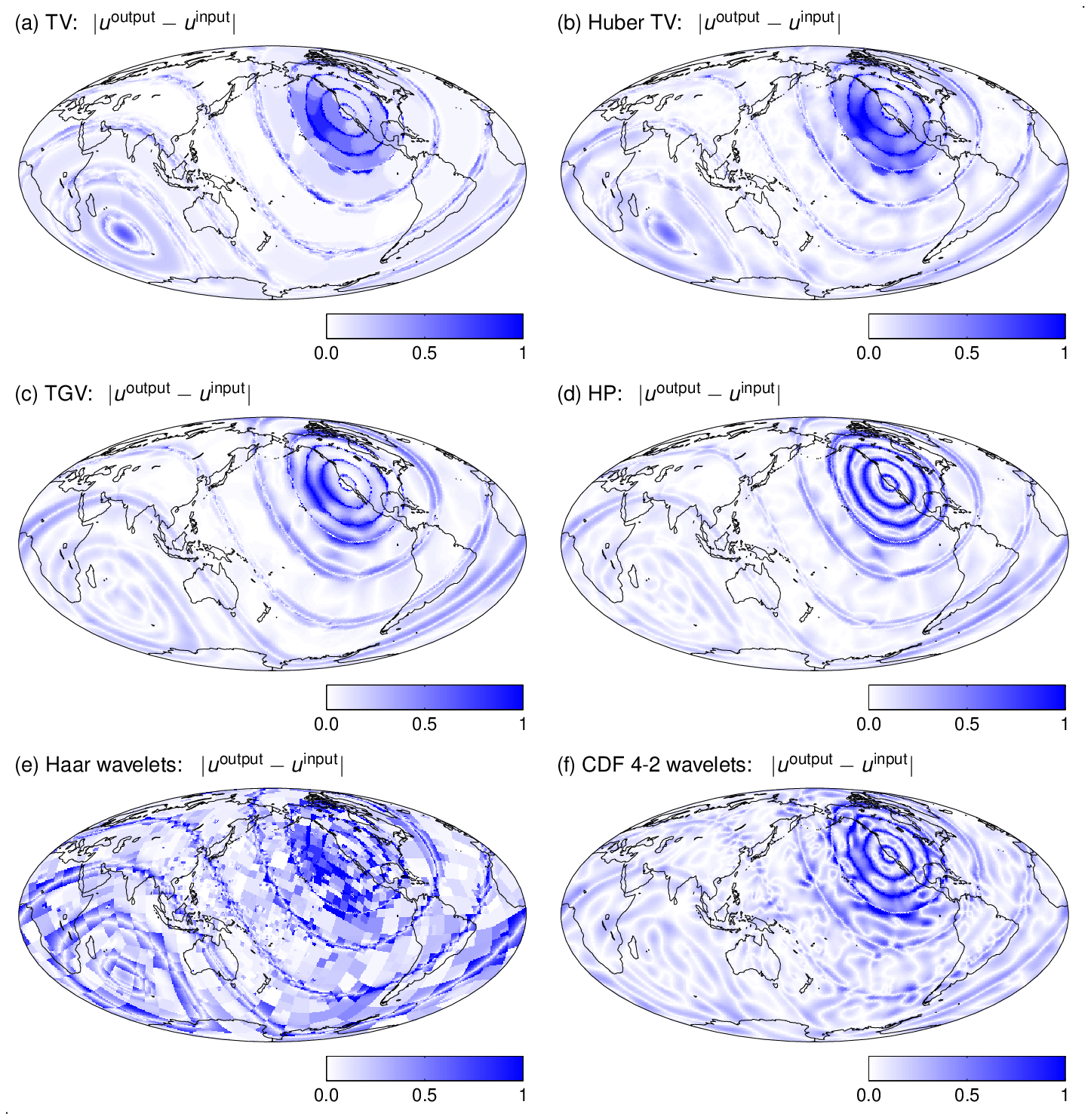}}
\caption{Map view of the differences $|u^\mathrm{output}-u^\mathrm{input}|$ for the six reconstructions. Large errors may occur near edges (as a result of smoothing or change of edge position in $u^\mathrm{output}$) and in other places (as a result of amplitude damping in $u^\mathrm{output}$).}\label{differencefigure}
\end{figure}

\begin{figure}
\centering\resizebox{\textwidth}{!}{\includegraphics{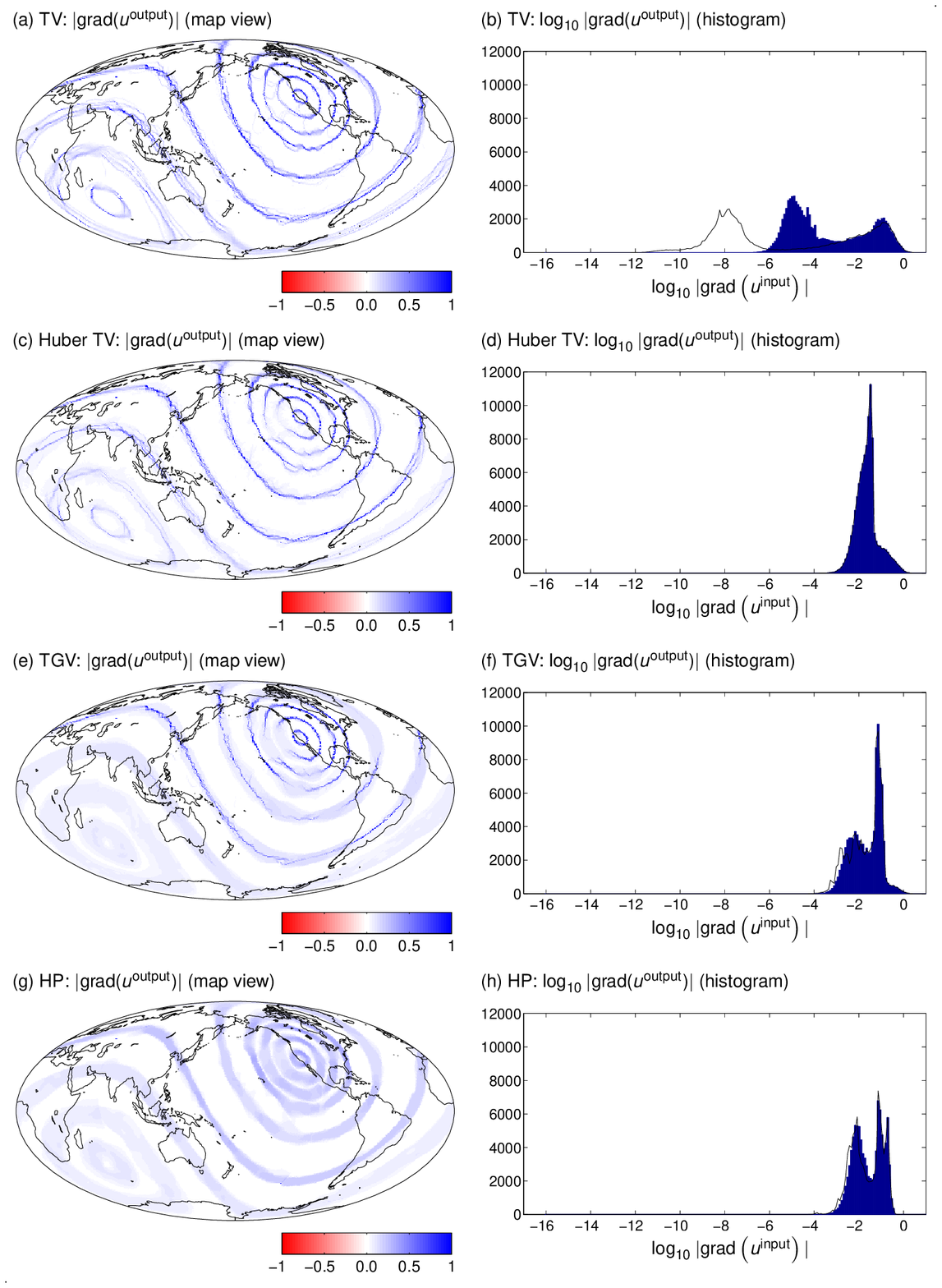}}
\caption{Gradient field of four reconstructed models of Section~\ref{reviewsection}. On the left hand side is a map view of the norm of the local gradient of the output models. On the right hand side is a histogram (of the logarithm) of the norm of the local gradient of these output models (in blue, after $1000$ iterations; a black line indicates the position of the histogram after $10^5$ iterations). For the first reconstruction (TV), one expects this gradient field to be zero in most places. The Huber-TV reconstruction penalizes the local gradient but small (nonzero) values remain. For the last two reconstructions one does not expect sparse (or almost sparse) gradients, but a piecewise constant gradient.}\label{gradfigure}
\end{figure}

The original input model possesses an azimuthal symmetry around the point $(36^{\circ} \mathrm{N}, -120^{\circ} \mathrm{E})$. However, the reconstructed models have lost this symmetry. This is due to the lack of data, the uneven distribution of the ray coverage, the noise on the data and the penalties imposed. For instance, the (circular) edges that are present in the input model are reconstructed, but the precise position is changed depending on this azimuthal angle.
Instead of showing a single cross section of the output models, as was done in the second column of Figures~\ref{reconstructionfigure} and \ref{waveletreconstructionfigure}, it makes sense to also calculate the cross sections along many great circles passing through the same point $(36^{\circ} \mathrm{N}, -120^{\circ} \mathrm{E})$ and average them. The average over $36$ great circles is displayed in Figure~\ref{averagecrosssectionfigure}; in the case of the TV reconstruction one sees that such an average cross section does not exhibit the same sharp edges as a single cross section does. The cause is simply that the output model is not circularly symmetric around this point. In other words the reconstructed edges do not necessarily coincide with the original edges. Averaging this effect leads to some smoothing. This also occurs in the Haar reconstruction. The effect is less evident on the other reconstruction that are already smooth. Figure~\ref{differencefigure} shows a map view of the difference $|u^\mathrm{output}-u^\mathrm{input}|$ between input model and reconstruction.

Figure~\ref{gradfigure} shows the length of the local gradient of four reconstructed models in map view and as a histogram. These should be compared to Figure~\ref{inputfigure}, panels (c) and (d). The gradient of the TV reconstructed model in panels (a) and (b) is sparse, although $|\mathrm{grad}(u^\mathrm{out})|$ is nowhere exactly equal to zero. This is a result of stopping the iterative reconstruction algorithm after a finite number of iterations (in this case after $1000$ iterations). This behavior is best seen on the histogram where the peak between $10^{-4}$ and $10^{-6}$ represents values of $|\mathrm{grad}(u^\mathrm{out})|$ that have not fully converged to zero. A black line shows the position of the histogram after $10^5$ iterations. We can clearly see that the secondary peak has moved to the left, i.e. to values around $10^{-8}$. The primary peak stays roughly in the same position, indicating little change to the significant values of the gradient after $1000$ iterations.

On the second row of Figure~\ref{gradfigure} we see that the Huber TV penalty does not lead to sparse gradients (for $\alpha\neq 0$). Although many values are small (colored in white on the map view), there is no heavy tail visible on the left hand side of the corresponding histogram in panel (d).

The TGV and HP penalties (panels (e)--(h)) give rise to local gradients that are markedly different in character from the TV and Huber TV cases. Here, as in the Huber TV case, we do not expect sparse gradients.  The HP reconstruction imposes a piece-wise linear solution, as much as the data allows for. Therefore the gradient of the output model will be piecewise constant. This is clearly observed in panel (e) of Figure~\ref{gradfigure}. In case of the TGV reconstruction, the results lie somewhere between the TV reconstruction and the HP reconstruction.

In case of the HP reconstruction, we expect a model that has sparse Hessian field (if one would plot the Frobenius norm of the local Hessian (see formulas (\ref{Hessdef}), (\ref{hessmatrix}) and (\ref{frobeniusdef}) for explicit expressions) the sparse nature of this Hessian field would be apparent). In other words, the second derivatives of the model will be mostly zero. This means that the output model will be piecewise linear.

The six imaging models (\ref{TVdef})--(\ref{waveletpenalty}) give rise to qualitatively different reconstructions. The associated minimization problems (\ref{penproblem}) or (\ref{conproblem}) however, can all be solved using the same iterative algorithms, with minor variations. The next section contains a description and numerical comparison of several suitable algorithms. Section~\ref{detailssection} contains the technical details on how these algorithms can be applied to the penalized problems discussed above.

\section{Algorithms}
\label{algsection}

In this section we review a number of iterative algorithms that can be used for solving problems (\ref{penproblem}) and (\ref{conproblem}) and make a numerical comparison of them. We start with some known algorithms for the penalized problem (\ref{penproblem}) in Subsection~\ref{penalgsection}.

In Subsection~\ref{conalg} we perform the same kind of comparison for two iterative algorithms for the constrained problem (\ref{conproblem}). One is a primal dual hybrid gradient algorithm of \cite{Esser.Zhang.ea2010} and the second one is, as far as the authors know, new. A proof of convergence of the second algorithm is therefore included in Section~\ref{proofsection}.

The algorithms presented generally involve matrix-vector multiplications and vector space operations (addition, multiplication with scalar). Furthermore they also involve a simple convex projection operator $P_\lambda$ or a simple soft-thresholding operation $S_\lambda$. These are defined componentwise by $P_\lambda((w_1,\ldots,w_N))=(P_\lambda(w_1),\ldots,P_\lambda(w_N))$ and $S_\lambda((w_1,\ldots,w_N))=(S_\lambda(w_1),\ldots,S_\lambda(w_N))$ with:
\begin{equation}
P_\lambda(w_i)=
\left\{
\begin{array}{ll}
\frac{w_i}{\|w_i\|}\lambda \quad&  \|w_i\|>\lambda\\
w_i & \|w_i\|\leq \lambda
\end{array}
\right.\qquad\mathrm{and}\qquad S_\lambda(w_i)=w_i-P_\lambda(w_i).\label{PSdef}
\end{equation}
Here $w_i$ may itself be an element of $\mathbb{R},\mathbb{R}^2,\ldots$. We refer to Section~\ref{detailssection} for some more information. The precise details may vary depending on wether one treats a 2D or 3D problem, or the precise form of the penalty (TV, Huber TV etc).

Below we will use the symbol $\|A\|^2$ to denote the largest eigenvalue of $A^TA$, and $\|K\|^2$ to denote the largest eigenvalue of $K^TK$. Apart from the model variable $u$, the iterative algorithms below also use one or several auxiliary variables $\bar u$, $w$ (which is the subgradient of $\lambda\|\cdot\|_1$), etc. The starting point of each of the algorithms in Subsections~\ref{penalgsection} and \ref{conalg} is arbitrary.

\subsection{Algorithms for penalized problems}
\label{penalgsection}

The generalized iterative soft-thresholding algorithm \citep{Loris.Verhoeven2011}:
\begin{equation}
\left\{
\begin{array}{lcl}
\bar u^{n+1}&=&u^n+\tau_1 K^T (y-Ku^n)-\tau_1 A^T w^n\\
w^{n+1}&=&P_\lambda\left(w^n+\frac{\tau_2}{\tau_1}A\bar u^{n+1}\right)\\
u^{n+1}&=&u^n+\tau_1 K^T (y-Ku^n)-\tau_1 A^T w^{n+1},
\end{array}
\right.\label{gista}
\end{equation}
converges to a minimizer of the penalized problem (\ref{penproblem}) if step sizes $\tau_1$ and $\tau_2$ are chosen as $\frac{1}{2}\tau_1\|K\|^2<1$ and $\tau_2\|A\|^2<1$.

When applied to problem (\ref{penproblem}) the algorithm of \cite{Chambolle.Pock2010} takes the form:
\begin{equation}
\left\{
\begin{array}{lcl}
w^{n+1}&=&P_\lambda\left(w^n+\sigma\frac{\tau_2}{\tau_1}A \bar u^{n}\right)\\
v^{n+1}&=&\frac{1}{1+\sigma}v^n+\frac{\sigma}{1+\sigma}(K\bar u^n-y)\\
u^{n+1}&=&u^n-\tau_1 K^Tv^{n+1}-\tau_1 A^T w^{n+1}\\
\bar u^{n+1}&=&u^{n+1}+\theta\left(u^{n+1}-u^n\right).
\end{array}
\right.\label{champock}
\end{equation}
It converges to a minimizer of problem (\ref{penproblem}) for $\theta=1$ and $\sigma\|\tau_1 K^TK+\tau_2 A^T A\|<1$.

An explicit Bregman algorithm:
\begin{equation}
\left\{
\begin{array}{lcl}
u^{n+1}&=&u^n+\tau_1 K^T (y-Ku^n)-\tau_1 A^T\left[w^n +\frac{1}{\theta}(w^n-w^{n-1})\right]\\
w^{n+1}&=&(1-\theta)w_n+\theta P_\lambda\left(w^n+\frac{\tau_2}{\tau_1}A u^{n+1}\right)\\
\end{array}
\right.\label{explbregman}
\end{equation}
converges for $0<\theta\leq 1$ and $\|\tau_1 K^TK+\tau_2
A^T A\|<1$ to a minimizer of problem (\ref{penproblem}). It is found from \cite[Equations 5.9, 5.10 and 5.11]{Zhang.Burger.ea2011}.

The conditions for convergence on $K$ and $A$ for algorithms (\ref{champock}) and (\ref{explbregman}) are different from the condition in algorithm (\ref{gista}). In the last two algorithms $K$ and $A$ are coupled in a single condition, whereas two separate conditions are used for algorithm (\ref{gista}). In other words, the conditions in algorithm (\ref{gista}) only depend on the matrix norms of $K$ and $A$, whereas in algorithms (\ref{champock}) and (\ref{explbregman}) the conditions also depend on the orientation of the singular vectors of $K$ with respect to the singular vectors of $A$. This makes algorithm (\ref{gista}) a little bit easier to use.
Also, no $\frac{1}{2}$ is found in front of $K^TK$in the conditions for algorithms (\ref{champock}) and (\ref{explbregman}).

It is also worth noting that for $\sigma=1$ and $\theta=1$ the algorithm (\ref{champock}) reduces to algorithm (\ref{explbregman}) for $\theta=1$. We omit the details of this calculation.

The above three algorithms are fully explicit: they only require the application of the matrix $K$ and its transpose at every step, and the application of the operator $A$ and its transpose. The nonlinear operator $P_\lambda$ also has a simple implementation (see also Section~\ref{detailssection} for explicit expressions).

In addition to the three explicit algorithms above, we also include the following implicit algorithm in our comparison (here implicit means that a linear system needs to be solved in every step):
\begin{equation}
\left\{
\begin{array}{lcl}
u^{n+1}&=&(K^TK+\alpha A^TA)^{-1}\left(K^Ty-A^T(w^n-\alpha z^n)\right)\\
z^{n+1}&=&S_{\lambda/\alpha}\left(Au^{n+1}+w^n/\alpha\right) \\
w^{n+1}&=&w^n+ \alpha (Au^{n+1}-z^{n+1})\\
\end{array}
\right.\label{implbregman}
\end{equation}
for $0<\theta\leq 1$ and a parameter $\alpha$. It is the so-called split-Bregman method  \citep{Goldstein.Osher2009}.

The split Bregman method was used in a seismic tomography context in \citep{Gholami.Siahkoohi2010}. The method seems appropriate for some special matrices $K$ and $A$, for which the inverse $(K^TK+\alpha A^TA)^{-1}$ is easy (such as e.g. the combination convolution $K$ and $A=$grad which can both be diagonalised in Fourier space) or when legacy code exists: If $K$ and $A$ have no special structure, the first line in the split-Bregman algorithm (\ref{implbregman}) itself needs an iterative algorithm.

The above four algorithms are compared in the framework of the TV reconstruction of Section~\ref{reviewsection}, i.e. using the matrix $K$ and the synthetic data $y$ of Secion~\ref{reviewsection} and the choice $A=$grad. A reference minimizer $u^\mathrm{ref}$ of problem (\ref{penproblem}) was first obtained using $10^5$ iterations of algorithm (\ref{gista}) with $\tau_1=1.99/\|K\|^2$ and $\tau_2=0.99/\|A\|^2$. This was done for two values of $\lambda$. One value of $\lambda$ yields a reconstructed model that fits the data to noise level $\|Ku^\mathrm{ref}-y\|/\|n\|\approx 1$, and another (larger) value of $\lambda$ that underfits the data: $\|Ku^\mathrm{ref}-y\|/\|n\|\approx 3$. In this way, it is possible to evaluate the performance of the algorithms for different values of $\lambda$. A larger value of $\lambda$ would be used for cases with more noise; a smaller value of $\lambda$ would be used for cases with less noise.

In Figure~\ref{penspeedfigure}, panels (a) and (b), the evolution of the relative error $\|u^n-u^\mathrm{ref}\|/\|u^\mathrm{ref}\|$ to the true minimizer is plotted as a function of computing time for these four algorithms. The algorithms ran for $1000$ iterations each (starting from zero model), except the split-Bregman algorithm which ran for only $200$ outer iterations (and $5$ inner iterations of conjugate gradient type). On the horizontal axis computing time is used rather than number of iterations as the split-Bregman algorithm's outer iteration step is more expensive than a single step of the explicit algorithms. Indeed, the explicit algorithms use a single application of the matrices $K$, $K^T$, $A$ and $A^T$ in each iteration whereas the implicit split-Bregman algorithm needs to solve a linear system in each step. The times mentioned in Figure~\ref{penspeedfigure} were obtained for this specific example on a single 2.66GHz CPU with 12GB memory running Matlab R2009a.  They cannot easily be extrapolated to other problems (with different numbers of variables, data, other matrices $K$ and $A$, different hardware). Panels (c) and (d) show the evolution of the functional (\ref{penproblem}) as a function of time.
For algorithms (\ref{gista}) and (\ref{champock}), one can show that the functional (\ref{penproblem}) tends to its limiting value as $\mathcal{O}(1/n)$ where $n$ is the number of iterations \citep{Loris.Verhoeven2011,Chambolle.Pock2010}. There is no such result for the error $\|u^{n}-\hat u\|$.

The algorithms that perform the best for the smaller value of $\lambda$ perform the worst for the larger value of $\lambda$. Moreover we see that the smallest value of $\|u^n-u^\mathrm{ref}\|/\|u^\mathrm{ref}\|$ does not necessarily correspond to the smallest value of $F(u^n)-F(u^\mathrm{ref})$. We conclude that all algorithms mentioned are suitable for solving problem (\ref{penproblem}), but that convergence may depend on a good choice of the step size parameters used. For algorithm (\ref{gista}) the larger step size $\tau_1=1.99/\|K\|^2$ is suitable for cases with low noise, whereas the choice $\tau_1=0.99/\|K\|^2$ is better in high noise cases. Generally, the error between $u^n$ and the true minimizer lies between 1 and 10\% after $1000$ iterations.

\begin{figure}
\centering\resizebox{\textwidth}{!}{\includegraphics{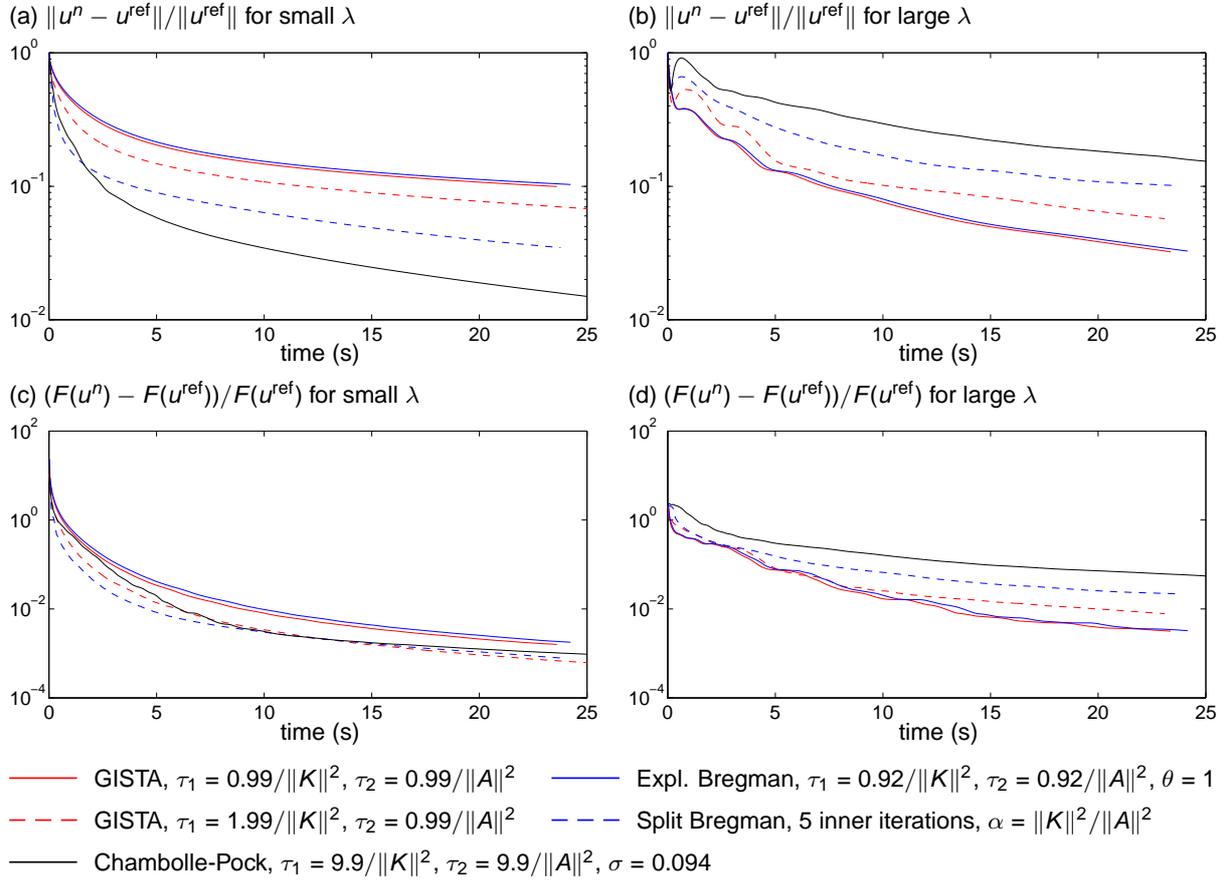}}
\caption{Convergence rate of the iterative minimization algorithms of Section~\ref{penalgsection}, with relative error to the true minimizer on the top and value of the functional on the bottom. The right hand side column refers to an experiment with a larger penalty parameter and left hand side column refers to an experiment with a smaller penalty parameter. The best algorithm depends on the value of the penalty parameter $\lambda$, on the step sizes used, and on the criteria (distance to true minimizer or value of functional). Computing time instead of number of iterations is used for fair comparison.}\label{penspeedfigure}
\end{figure}

It was already remarked that these algorithms do not give sparse $Au^n$ at every iteration step; only for the limiting value of $u^\infty$ will $Au^\infty$ be sparse. When $A$ is the unit matrix or an orthogonal matrix, then the generalized iterative soft threshdolding algorithm (\ref{gista}) reduces to the traditional soft-thresholding algorithm of \cite{Daubechies2004b} which does produce sparse $u^n$ at every step:
\begin{equation}
u^{n+1}=S_{\lambda\tau_1}\left(u^n+\tau_1 K^T (y-Ku^n)\right)\quad\mathrm{for}\quad
\min_u\frac{1}{2}\|Ku-y\|^2+\lambda\|u\|_1\label{ista}
\end{equation}
for step sizes $\tau_1\|K\|^2<2$. Such a simplified algorithm could be used for sparse recovery in a wavelet basis as was done in \citep{Loris.Nolet.ea2007} in a seismic tomography context. It suffices to make the change of variables $u=W^{-1}w$ in expressions (\ref{penproblem}) and (\ref{waveletpenalty}) to use algorithm (\ref{ista}) with $K$ replaced by $KW^{-1}$ and $u^{n}$ replaced by $w^n$. Such a change of variables is not possible for the TV penalty as $A=\mathrm{grad}$ is not invertible. An accelerated (more efficient) version of algorithm (\ref{ista}), the so-called Fast Iterative Soft-Thresholding Algorithm (FISTA), is described in \citep{Beck.Teboulle2008} (for $\tau_1\|K\|^2<1$). See also \cite{Yamagishi2011} for some recent developments.

For $K=1$ (or orthogonal) and $\tau_1=1$ the algorithm (\ref{gista}) reduces to the following projected gradient algorithm:
\begin{equation}
\left\{
\begin{array}{lcl}
w^{n+1}&=&P_\lambda\left(w^n+\tau_2 A (K^T y- A^T w^n) \right)\\
u^{n+1}&=& K^T y- A^T w^{n+1}
\end{array}
\right.
\quad\mathrm{for}\quad
\min_u\frac{1}{2}\|u-y\|^2+\lambda\|Au\|_1
\label{projgrad}
\end{equation}
for step size $\tau_2\|A\|^2<2$ \citep{Chambolle2005}.
This algorithm can also be accelerated \citep{Nesterov1983a}. The case $K=1$ and $A=\mathrm{grad}$ corresponds to denoising with a total variation penalty. A recent comprehensive comparison of various numerical algorithms for this task can be found in Section 6.2.1 of \citep{Chambolle.Pock2010}.

\subsection{Algorithms for constrained problems}
\label{conalg}

In this subsection we compare two explicit iterative algorithms for the constrained minimization problem (\ref{conproblem}). This problem depends on the parameter $\epsilon$ which determines the desired data misfit: $\|Kx-y\|\leq \epsilon$. Such a formulation of a reconstruction problem is therefore useful when one wants to fit the data to noise level. It then suffices to set the parameter $\epsilon$ equal to the norm of the noise vector (Morozov discrepancy principle).

In addition to the projection operator $P_\lambda$ the constrained algorithms below also use another projection operator $Q_{y,\epsilon}$ defined as:
\begin{equation}
Q_{y,\epsilon}(v)=
\left\{
\begin{array}{lrl}
y+\epsilon \frac{v-y}{\|v-y\|} & \quad\mathrm{for}\quad & \|v-y\|> \epsilon\\
v & \quad\mathrm{for}\quad & \|v-y\|\leq \epsilon
\end{array}
\right.\qquad\mathrm{and}\qquad T_{y,\epsilon}(v)=v-Q_{y,\epsilon}(v)
\end{equation}
given $y$ and $\epsilon$. In other words, $Q_{y,\epsilon}$ is the projection on the $\ell_2$ ball of radius $\epsilon$ centered at $y$ and $T_{y,\epsilon}$ is an associated thresholding function.

We start with the so-called primal-dual hybrid gradient algorithm (PDHGMp) found in equation 5.3 of \citep{Esser.Zhang.ea2010}. In the present notation this algorithm takes the form:
\begin{equation}
\left\{
\begin{array}{lcl}
u^{n+1}&=&u^n-\tau_1 K^T(v^n+(v^n-v^{n-1})/\theta)-\tau_1 A^T (w^n+(w^n-w^{n-1})/\theta)\\
w^{n+1}&=&P_{\mu/\tau_1}\left(w^n+\frac{\tau_2}{\tau_1}A u^{n+1}\right) \\
v^{n+1}&=&(1-\theta)v^n+\theta T_{y,\epsilon}\left(v^n+ K u^{n+1}\right)
\end{array}
\right.\label{PDHGmPalg}
\end{equation}
and converges for $0<\theta\leq 1$, $\mu>0$, $\|\tau_1 K^T K+\tau_2 A^T A\|<1$. This algorithm can also be derived from algorithm $(A_1)$ on page $28$ of \citep{Zhang.Burger.ea2011} and (for $\theta=1$) from algorithm 1 in \citep{Chambolle.Pock2010}. We omit the details.

Another algorithm is given by the formulas:
\begin{equation}
\left\{
\begin{array}{lcl}
\bar u^{n+1}&=&u^n-\tau_1 K^T(v^n+(v^n-v^{n-1})/\theta)-\tau_1 A^T w^n\\
w^{n+1}&=&P_{\mu/\tau_1}\left(w^n+\frac{\tau_2}{\tau_1} A\bar u^{n+1}\right) \\
u^{n+1}&=&u^n-\tau_1 K^T(v^n+(v^n-v^{n-1})/\theta)-\tau_1 A^T w^{n+1}\\
v^{n+1}&=&(1-\theta)v^n+\theta T_{y,\epsilon}\left(v^n+ K u^{n+1}\right)
\end{array}
\right.\label{combialg}
\end{equation}
and it is proven in Section~\ref{proofsection} that it converges to the minimizer of problem (\ref{conproblem}) for $0<\theta\leq 1$, $\mu>0$, $\tau_1\|K\|^2<1$ and $\tau_2\|A\|^2<1$. We will refer to this algorithm as the generalized basis pursuit denoising algorithm (GBPDNA). The motivation for this name is given at the end of this Section.

In Figure~\ref{conspeedfigure} the two algorithms are compared numerically on the same problem as in Subsection~\ref{penalgsection} (i.e. same $K$, $A$ and $y$). Moreover, we chose two values of $\epsilon$ corresponding to the value of $\|Ku^\mathrm{ref}-y\|$ taken by the two reference minimizers of the simulation in the Subsection~\ref{penalgsection}. In other words, we chose $\epsilon$ so that the corresponding minimizers are identical to the minimizers of the numerical simulation in Subsection~\ref{penalgsection}. Because of this choice of $\epsilon$ it is possible to directly compare panels (a) and (b) of Figure~\ref{conspeedfigure} with panels (a) and (b) of Figure~\ref{penspeedfigure}.

We recompute the true (reference) minimizer $u^\mathrm{ref}$ with $10^5$ iterations of algorithm (\ref{combialg}) with $\tau_1=0.99/\|K\|^2$, $\tau_2=0.99/\|A\|^2$ and $\mu=\|K^Ty\|_\infty$. We verified that the resulting minimizer is equal (up to $0.1\%$) to the reference minimizer of subsection~\ref{penalgsection}. Then we compare the iterates of algorithms (\ref{PDHGmPalg}) and (\ref{combialg}) to these reference minimizers. We clearly notice that both algorithms are almost identical in their convergence behavior. The effect of an appropriate choice of the parameter $\mu$ in these algorithms is also visible in Figure~\ref{conspeedfigure}. We conclude that there is no significant difference in convergence speed between these two algorithms (\ref{PDHGmPalg}) and (\ref{combialg}) for the constrained minimization problem (\ref{conproblem}). In practice it is slightly easier to work with two conditions of type $\tau_1\|K\|^2<1$ and $\tau_2\|A\|^2<1$ than one condition of type $\|\tau_1 K^T K+\tau_2 A^T A\|<1$.

\begin{figure}
\centering\resizebox{\textwidth}{!}{\includegraphics{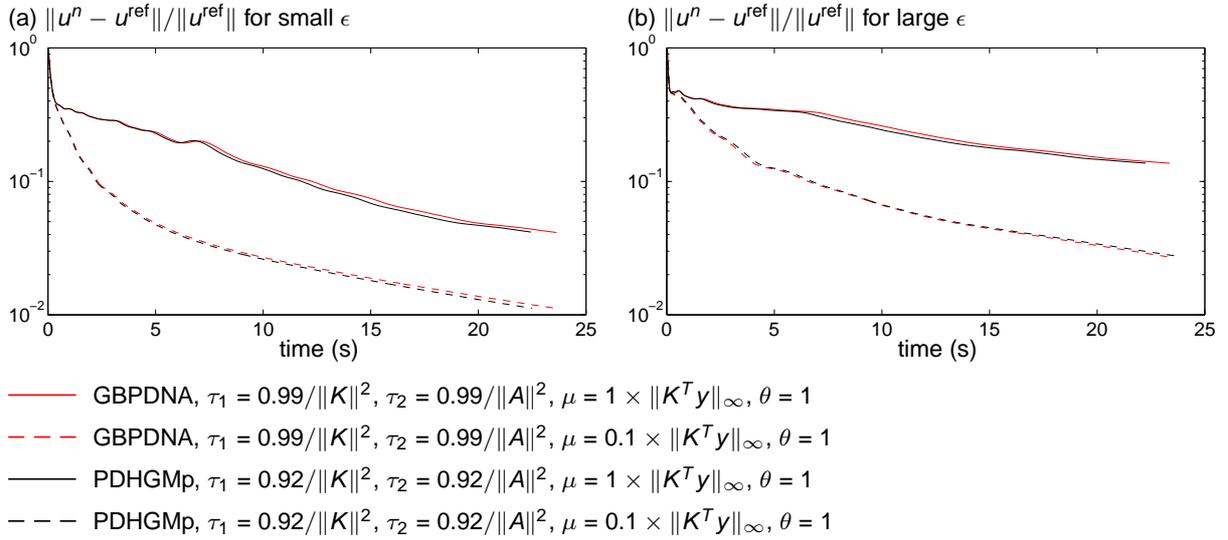}}
\caption{Convergence rate of two iterative minimization algorithms for the constrained problem (\ref{conproblem}). Pictured is the relative error to the true minimizer (obtained from algorithm (\ref{combialg}) and $10^5$ iterations). The figure on the left hand side is for a small value of $\epsilon$ and the one on the right is for a larger value of $\epsilon$. Each algorithm was run for $1000$ iterations.}\label{conspeedfigure}
\end{figure}

The constrained problem (\ref{conproblem}) reduces to the so-called `basis pursuit denoising' problem:
\begin{equation}
\arg\min_{\|Ku-y\|\leq\epsilon}\|u\|_1\label{bpproblems1}
\end{equation}
when
$A=1$, and to `basis pursuit' \citep{Chen.Donoho.ea1998}
\begin{equation}
\arg\min_{Ku=y}\|u\|_1\label{bpproblems2}
\end{equation}
when $A=1$ and $\epsilon=0$.
In the case of problem (\ref{bpproblems1}) the proposed algorithm (\ref{combialg}) reduces to:
\begin{equation}
\left\{
\begin{array}{lcl}
u^{n+1}&=&S_\mu(u^n-\tau_1 K^T(v^n+(v^n-v^{n-1})/\theta))\\
v^{n+1}&=&(1-\theta)v^n+\theta T_{y,\epsilon}\left(v^n+ K u^{n+1}\right)
\end{array}
\right. \label{bpdnalg}
\end{equation}
and in case of problem (\ref{bpproblems2}) the proposed algorithm (\ref{combialg}) reduces to
\begin{equation}
\left\{
\begin{array}{lcl}
u^{n+1}&=&S_\mu(u^n-\tau_1 K^T(v^n+(v^n-v^{n-1})/\theta))\\
v^{n+1}&=&(1-\theta)v^n+\theta (v^n+ K u^{n+1}-y).
\end{array}
\right.\label{bpalg}
\end{equation}
The last special case (\ref{bpalg}) is the same as algorithm 5.6 of \citep{Zhang.Burger.ea2011}.
For this reason we will call algorithm (\ref{combialg}) the generalized basis pursuit denoising algorithm (GBPDNA).

In algorithms (\ref{bpdnalg}) and (\ref{bpalg}) the step size parameter $\tau_1$ satisfies $\tau_1\|K\|^2<1$. As a result of the soft-thresholding, the $u^n$ are sparse in every step.

\section{Explicit formulas for the algorithms used in Section~\ref{reviewsection}}
\label{detailssection}

The iterative algorithms of Sections~\ref{penalgsection} and \ref{conalg} were compared in the framework of TV -penalized seismic recovery (equations (\ref{penproblem}) and (\ref{TVdef})). The same algorithms, with minor modifications, can also be used to solve the other three problems of Section~\ref{reviewsection} as well. In this section we give some explicit formulas for the expressions encountered in these penalties and algorithms. The details depend e.g. on the number of spatial dimensions (2D, or 3D), on the precise expression of the differencing matrix $A$ in functionals (\ref{penproblem}) and (\ref{conproblem}), etc.

For the TV regularized problem (\ref{TVdef}) we set $A=\mathrm{grad}$, which we can define as $\mathrm{grad}(u)=(\Delta_x u,\Delta_y u)$ and :
\begin{equation}
(\Delta_x u)_{i,j}=
\left\{
\begin{array}{lcl}
u_{i+1,j}-u_{i,j} & & i:1\ldots N-1\\
0 & & i=N
\end{array}\right.
\quad
(\Delta_y u)_{i,j}=
\left\{
\begin{array}{lcl}
u_{i,j+1}-u_{i,j} & & j:1\ldots N-1\\
0 & & j=N
\end{array}\right.
\label{graddef}
\end{equation}
($i,j:1\ldots N$) for $u\in \mathbb{R}^{N\times N}$.
So for a 2D model $u\in \mathbb{R}^{N\times N}$, the gradient field $\mathrm{grad}(u)$ will be in $\mathbb{R}^{N\times N\times 2}$.
The transpose of $A$ is then given by the formulas $A^T (w_x,w_y)=\Delta_x^T w_x+\Delta_y^T w_y$ with:
\begin{equation}
\begin{array}{l}
(\Delta_x^T w_x)_{i,j}=
\left\{
\begin{array}{lcl}
-w_{x,i,j}+w_{x,i-1,j} \qquad& & i:2\ldots N-1\\
-w_{x,i,j} & & i=1\\
w_{x,i-1,j} & & i=N\\
\end{array}\right.\\
\mathrm{and} \\
(\Delta_y^T w_y)_{i,j}=
\left\{
\begin{array}{lcl}
-w_{y,i,j}+w_{y,i,j-1} \qquad & & j:2\ldots N-1\\
-w_{y,i,j} & & j=1\\
w_{y,i,j-1} & & j=N\\
\end{array}\right.
\end{array}\label{divdef}
\end{equation}
($i,j:1\ldots N$) for a $w_x,w_y\in\mathbb{R}^{N\times N}$. With these definitions one has that $\langle Au,w\rangle=\langle u,A^T w\rangle$ for all $u\in \mathbb{R}^{N\times N}$ and all $w\in \mathbb{R}^{N\times N\times 2}$.

These operations are easy to code (e.g. in MATLAB) and the extension of the above formulas to 3D models ($u\in\mathbb{R}^{N\times N\times N}$) is straightforward. In the examples of Section~\ref{reviewsection} we used the cubed sphere parametrization of \cite{Ronchi.Iacono.ea1996}. The formulas for $\mathrm{grad}$ (or $\Delta_x$ and $\Delta_y$) and $\mathrm{grad}^T$ (or $\Delta_x^T$ and $\Delta_y^T$) are the same as in (\ref{graddef}) and (\ref{divdef}), except that other boundary conditions are used to ensure the correct behavior at the edges of the six (`square') faces that make up the parametrization of the sphere.

For minimizing the functional (\ref{TVdef}) the generalized iterative soft-thresholding algorithm (\ref{gista}) was used. The nonlinear operator $P_\lambda$ appearing in it, in this case becomes in accordance with formula (\ref{PSdef}):
\begin{equation}
P_\lambda(w_x,w_y)=
\left\{
\begin{array}{lcl}
\displaystyle \frac{\lambda}{\sqrt{w_x^2+w_y^2}}(w_x,w_y) &  \qquad &\sqrt{w_x^2+w_y^2}>\lambda\\
(w_x,w_y) & &\sqrt{w_x^2+w_y^2}\leq \lambda,
\end{array}
\right.
\end{equation}
for $(w_x,w_y)\in\mathbb{R}^2$. Here we have dropped the double subscript $i,j\in\{1\ldots N\}$ for clarity.

For the Huber TV regularization method, one uses the same expressions for $A=\mathrm{grad}$ and $A^T$ as in the TV case, but the operator $P_\lambda$ has to be replaced by:
\begin{equation}
P_{\lambda,\alpha}(w_x,w_y)=
\left\{
\begin{array}{lcl}
\displaystyle \frac{\lambda}{\sqrt{w_x^2+w_y^2}}(w_x,w_y) & \qquad & \sqrt{w_x^2+w_y^2}>\lambda+\alpha\\
\displaystyle\frac{\lambda}{\lambda+\alpha}(w_x,w_y) &\qquad & \sqrt{w_x^2+w_y^2}\leq \lambda+\alpha,
\end{array}
\right.\label{TVPlambdadef}
\end{equation}
again applied in every pixel.
With this minor modification the algorithms of Section~\ref{algsection} may also be applied to problem (\ref{HTVdef}). The convergence is still guaranteed under the same conditions as in Theorem~\ref{theorem} (see also note after proof in Section~\ref{proofsection}).

For the problem (\ref{Hessdef}), we have set
$Au=\mathrm{Hess}(u)$ with
\begin{equation}
\mathrm{Hess}(u)=\left(
\begin{array}{cc}
\Delta_x^2u & \Delta_x\Delta_y u\\
\Delta_y\Delta_xu & \Delta_y^2 u
\end{array}\right)\label{hessmatrix}
\end{equation}
for $u\in\mathbb{R}^{N\times N}$, i.e. $\mathrm{Hess}(u)\in\mathbb{R}^{N\times N \times 2\times 2}$. Its transpose $A^T$ is given by the formula:
\begin{equation}
\mathrm{Hess}^T(w)=(\Delta_x^T)^2 w_{11} +\Delta_y^T \Delta_x^T w_{12} +
\Delta_x^T \Delta_y^T w_{21} + (\Delta_y^T)^2 w_{22}
\end{equation}
acting on a $w\in\mathbb{R}^{N\times N\times 2\times 2}$, with $\Delta_x^T$ and $\Delta_y^T$ defined as in (\ref{divdef}).
The Frobenius norm of the local Hessian, as used in penalty (\ref{Hessdef}), is:
\begin{equation}
\left\|\left(
\begin{array}{cc}
H_{11} & H_{12}\\
H_{21} & H_{22}
\end{array}
\right)\right\|_F
=\sqrt{H_{11}^2+H_{12}^2+H_{21}^2+H_{22}^2}\label{frobeniusdef}
\end{equation}
for $H\in\mathbb{R}^{2\times 2}$. It is equal to $\sqrt{\sigma_1^2+\sigma_2^2}$ where $\sigma_1,\sigma_2$ are the singular values of $H$.
In formula (\ref{Hessdef}) the sum over all pixels of expression (\ref{frobeniusdef}) is taken, i.e.:
\begin{equation}
\lambda\|Ax\|_1=\lambda\sum_\mathrm{pixels}
\left\|\left(
\begin{array}{cc}
\Delta_x^2u & \Delta_x\Delta_y u\\
\Delta_y\Delta_xu & \Delta_y^2 u
\end{array}\right)\right\|_F.
\end{equation}
The advantage of using this matrix norm lies not only in the fact that it is easy to compute but also in the fact that the penalty becomes isotropic. In other words, a rotation of the coordinate axis, does not change this expression (the rotation $(x',y')=R(x,y)$ leads to $H'=R H R^T$ and therefore $\|H'\|_F=\|H\|_F$).
Other spectral matrix norms (i.e. norms that are based on the spectrum of $H$), besides the Frobenius norm, also lead to isotropic penalties. However, the Frobenius norm is particularly easy to use as it does not require an explicit singular value decomposition of each $H$ (in each pixel) to compute it. And, as we shall see, the associated projection $P_\lambda$ is also easy to compute.
Note that e.g. the penalty $\sum_{ij}|(\Delta u)_{ij}|$ is also isotropic. It leads to models where the laplacian $\Delta u$ of $u$ will be mostly zero, i.e. models that are piecewise harmonic (not piecewise linear).

In the case (\ref{Hessdef}) the expression for $P_\lambda$ used in algorithm (\ref{gista}) is
\begin{equation}
P_\lambda
\left(
\begin{array}{cc}
w_{11} & w_{12}\\
w_{21} & w_{22}
\end{array}
\right)
=
\left\{
\begin{array}{lcl}
\lambda\left(
\begin{array}{cc}
w_{11} & w_{12}\\
w_{21} & w_{22}
\end{array}
\right) /
\left\|
\left(
\begin{array}{cc}
w_{11} & w_{12}\\
w_{21} & w_{22}
\end{array}
\right)
\right\|_F
&  \qquad &\left\|
\left(
\begin{array}{cc}
w_{11} & w_{12}\\
w_{21} & w_{22}
\end{array}
\right)
\right\|_F>\lambda\\[4mm]
\left(
\begin{array}{cc}
w_{11} & w_{12}\\
w_{21} & w_{22}
\end{array}
\right)
 & &
\left\|
\left(
\begin{array}{cc}
w_{11} & w_{12}\\
w_{21} & w_{22}
\end{array}
\right)
\right\|_F\leq \lambda,
\end{array}
\right.\label{HessPlambdadef}
\end{equation}
for $(w_{11},w_{12};w_{21},w_{22})\in\mathbb{R}^{2\times 2}$. Again we have dropped the double subscript ${i,j}$ for clarity but it is understood that the above formula should be applied in every pixel. As already mentioned, this expression for $P_\lambda$ does not require a singular value decomposition of $H$ in every pixel.

In case of the total generalized variation penalty (\ref{GTVdef}), one has a functional of the form $\frac{1}{2}\|Ku-y\|^2+\lambda(\|\mathrm{grad}(u)-v\|_1+\alpha \|Dv\|_1)$. It can therefore be treated with a combination of the above formulas for $P_\lambda, \mathrm{grad},\Delta_x,\Delta_y$, etc. The matrix $A$ in algorithm (\ref{gista}) now takes the form $\left(
\begin{array}{cc}
\mathrm{grad} & -\mathrm{Id}\\
0 & \alpha D
\end{array}
\right)$
and acts on $(u,v)$ where $v$ is an auxiliary variable in $\mathbb{R}^{N\times N\times 2}$. The operator $D$ acting on $v$ is a derivative: $Dv=\left(
\begin{array}{cc}
\Delta_x v_x & \Delta_y v_x\\
\Delta_x v_y & \Delta_y v_y
\end{array}
\right)$. The explicit expression for penalty (\ref{GTVdef}) involves therefore:
\begin{equation}
\sum_\mathrm{pixels}\sqrt{(\Delta_x u-v_x)^2+(\Delta_y u-v_y)^2}+\alpha
\sqrt{(\Delta_x v_x)^2+(\Delta_y v_x)^2+(\Delta_x v_y)^2+(\Delta_y v_y)^2}.
\label{tgvexplicit}
\end{equation}
This penalty is denoted by $\neg\mathrm{symTGV}$ in \citep{Bredies.Kunisch.ea2010}.
The explicit expression for the nonlinear projection operator $P_\lambda$ used in the iterative algorithms now combines both expressions (\ref{TVPlambdadef}) (for the first line of $A$) and (\ref{HessPlambdadef}) (for the second line of $A$). In this case the norm $\|A\|$ depends on the parameter $\alpha$.

For the simulation of Section~\ref{reviewsection}, the algorithm (\ref{gista}) was implemented four times with the above definitions of $A$ and $P_\lambda$. For those cases, we chose $\tau_1=1.99/\|K\|^2$ and $\tau_2=0.99/\|A\|^2$ and $1000$ iterations were performed. $1000$ iterations correspond to about $20$ seconds of computer time on a single 2.66GHz CPU with 12GB of memory, running Matlab R2009a. The relative error with the input model, $\|x^\mathrm{output}-x^\mathrm{input}\|/\|x^\mathrm{input}\|$, is given in Table~\ref{results} for each of the six reconstructions. Two sparse wavelet reconstruction were also calculated in Section~\ref{reviewsection}. Here $1000$ iterations of the FISTA algorithm \citep{Beck.Teboulle2008} were used with appropriate wavelet transform $W$ (see formula (\ref{waveletpenalty}) and algorithm (\ref{ista}) with $K\rightarrow KW^{-1}$).

None of the six output models $u^\mathrm{output}$ is expected to be exactly equal to the input model $u^\mathrm{input}$ because of the lack of data, the noise on the data and the effect of the penalties on the reconstructions.

The computational complexity of the algorithms depends on the matrices $K$ and $A$ (or $W$). It was already mentioned in Section~\ref{algsection} that the iterative algorithms discussed in this paper use a single application of $K$, $K^T$, $A$ and $A^T$ in every step (except for algorithm (\ref{implbregman})). The matrix $K$ encodes the relationship between the model $u$ and data $y$. In case $K$ is a dense matrix, a single application of $K$ (or its transpose) requires $\mathcal{O}(mN)$ operations. Here $N$ is the number of components of $u$ and $m$ is the number of data. If the matrix $K$ is sparse (as is the case for the examples in this paper), or has structure (as e.g. in deconvolution) then this can be reduced significantly. If the matrix $A$ is chosen as a local differencing operator, then applying $A$ (or its transpose) requires $\mathcal{O}(N)$ operations. The same is true for a wavelet transform $W$. The convex projections $P_\lambda$ are also $\mathcal{O}(N)$ operations (as they are applied componentwise). The operator $T_{y,\epsilon}$ used in the constrained algorithms (\ref{PDHGmPalg}) and (\ref{combialg}) has computational complexity $\mathcal{O}(m)$.
The computing times mentioned in this paper do not serve as a basis for extrapolation to other situations; what is important here is that the application of such $A$ and $P_\lambda$ typically takes less time than applying $K$ and $K^T$. This is what makes the application of the edge-preserving penalties of Section~\ref{reviewsection} possible in practice.

\begin{table}\centering
\begin{tabular}{lllll}
Name & $\|u^\mathrm{output}-u^\mathrm{input}\|/\|u^\mathrm{input}\|$ & $\|Ku^\mathrm{output}-y\|/\|n\|$  \\ \hline
TV &  0.20643 & 1.0009\\
Huber-TV & 0.22519 & 1.0041 \\
TGV &  0.20843 & 1.0003 \\
HP & 0.20033 & 1.0080 \\
Haar & 0.39789 & 1.0034&\\
CDF 4-2 & 0.25248 & 1.0059&
\end{tabular}
\caption{Reconstruction results for the six reconstructions of Section~\ref{reviewsection}. All reconstructions fit the data equally well (third column), and have slightly different reconstruction error (second column).}\label{results}
\end{table}

\section{Proof of convergence of constrained algorithm (\ref{combialg})}
\label{proofsection}

Without loss of generality we set $\tau_1=\tau_2=1$ in algorithm (\ref{combialg}) and prove convergence for $\|K\|<1$ and $\|A\|<1$ (the step sizes can be introduced by scaling the matrices $K$, $A$ and the data $y$). The algorithm (\ref{combialg}) can thus  be written as:
\begin{equation}
\left\{
\begin{array}{lcl}
\bar u^{n+1}&=&u^n-K^T(v^n+Ku^n-z^n)-A^T w^n\\
w^{n+1}&=&P_\mu(w^n+ A\bar u^{n+1}) \\
u^{n+1}&=&u^n-K^T(v^n+Ku^n-z^n)-A^T w^{n+1}\\
z^{n+1}&=&Q_{y,\epsilon}(Ku^{n+1}+v^n)\\
v^{n+1}&=&v^n+\theta (Ku^{n+1}-z^{n+1}).
\end{array}
\right.\label{combinedfull}
\end{equation}
with $\|K\|<1$ and $\|A\|<1$.
Indeed, the variable $z^{n+1}$ can be eliminated from the last line of (\ref{combinedfull}) to yield the last line of (\ref{combialg}):
\begin{displaymath}
\begin{array}{lcl}
v^{n+1}&=&v^n+\theta \left(Ku^{n+1}-z^{n+1}\right)\\
&=& (1-\theta)v^n+\theta \left(Ku^{n+1}+v^n-Q_{y,\epsilon}\left(Ku^{n+1}+v^n\right)\right)\\
&=& (1-\theta)v^n+\theta T_{y,\epsilon}\left(Ku^{n+1}+v^n\right)
\end{array}
\end{displaymath}
where we used $T_{y,\epsilon}(a)=a-Q_{y,\epsilon}(a)$.
Substitution of the last line of (\ref{combinedfull}) (with $n\rightarrow n-1$) in the first line of algorithm (\ref{combinedfull}) then yields the first line of algorithm (\ref{combialg}).

The minimizer of problem (\ref{conproblem}) is determined by its variational equations. These are derived as follows. Introducing a positive parameter $\mu$, we write problem (\ref{conproblem}) as $\min_u \mu\|Au\|_1+I(Ku)$, where $I$ is the indicator function of the $\ell_2$ ball of radius $\epsilon$ around $y$: $I(Ku)=0$ for $\|Ku-y\|\leq\epsilon$ and $I(Ku)=\infty$ for $\|Ku-y\|>\epsilon$. The variational equations of the constrained problem (\ref{conproblem}) are therefore
\begin{equation}
A^T w+K^T v=0,
\end{equation}
with $w$ an element of the subdifferential of $\mu\|\cdot\|_1$ at $Au$ and $v$ an element of the subdifferential of $I$ at $Ku$.

The subdifferential $w$ of $\mu\|Ax\|_1$ satisfies $w_i=\mu\,(Au)_i/|(Au)_i|$ for $(Au)_i\neq 0$ and $|w_i|\leq\mu$ for $(Au)_i=0$. This means that $(Au)_i=S_\mu(w_i+(Au)_i)$ or equivalently, using (\ref{PSdef}), that $w_i=P_\mu(w_i+(Au)_i)$, which we write as $w=P_\mu(w+Au)$.

Similarly one shows that the subdifferential of $v$ of $I$ at $Ku$ is characterized by the relation $Ku=Q_{y,\epsilon}(Ku+v)$.
The variational equations that determine the minimizer of problem (\ref{conproblem}) can therefore be written with an auxiliary variable $z=Ku$ as:
\begin{equation}
\left\{
\begin{array}{lcl}
w&=&P_\mu\left(w+ A u\right) \\
u&=&u-K^T(v+Ku-z)-A^T w\\
z&=&Q_{y,\epsilon}\left(Ku+v\right)\\
v&=&v+\theta \left(Ku-z\right),
\end{array}
\right.\label{fixedpoint}
\end{equation}
which correspond to the fixed point equations of iteration (\ref{combinedfull}) when $\theta\neq 0$.

\begin{lemma}\label{lemma} Let $P_C$ be a projection on a non-empty closed convex set $C\subseteq\mathbb{R}^N$. Let $u^+,u^-,\Delta\in\mathbb{R}^N$. If $u^+=P_C(u^-+\Delta)$, then
\begin{equation}
\|u^+-u\|^2\leq\|u^--u\|^2-\|u^+-u^-\|^2 -2\langle u-u^+,\Delta\rangle\label{lemmaeq}
\end{equation}
for all $u\in C$.
\end{lemma}
\begin{proof}
Because $P_C$ is the projection on the convex set $C$ we have:
\begin{displaymath}
\langle u-P(u'),u'-P_C(u')\rangle\leq 0
\end{displaymath}
for all $u\in C$ and all $u'$. Setting $u'=u^-+\Delta$ and $P_C(u')=u^+$ in this inequality yields
\begin{displaymath}
\langle u-u^+,u^-+\Delta-u^+\rangle\leq 0.
\end{displaymath}
As $\langle u-u^+,u^--u^+\rangle=\left(\|u-u^+\|^2+\|u^--u^+\|^2-\|u-u^-\|^2\right)/2$ we find the relation (\ref{lemmaeq}).
\end{proof}

The proof of the following theorem is a combination of the proof of theorem~2 in \citep{Loris.Verhoeven2011} and theorem~4.2 of \citep{Zhang.Burger.ea2011}.
\begin{theorem}\label{theorem} In a finite dimensional setting, and when $\|K\|<1$, $\|A\|<1$ and $0<\theta\leq 1$, the iteration (\ref{combinedfull}) converges to a fixed point and provides a minimizer of problem (\ref{conproblem}) .
\end{theorem}
\begin{proof}
It was already remarked  that the fixed point equations (\ref{fixedpoint}) correspond to the variational equations for problem (\ref{conproblem}). The problem (\ref{conproblem}) has a solution and therefore a solution $(\hat u,\hat w,\hat v,\hat z)$ to the equations (\ref{fixedpoint}) exists.

Taking into account line 2 of algorithm (\ref{combinedfull}) and lemma~\ref{lemma} with $u^+=w^{n+1}$, $u^-=w^n$, $\Delta=A\bar u^{n+1}$, $u=\hat w$, and $P_C=P_\mu$, we find that:
\begin{displaymath}
\|w^{n+1}-\hat w\|^2\leq\|w^{n}-\hat w\|^2-\|w^{n+1}-w^n\|^2-2\langle \hat w-w^{n+1},A\left(u^{n+1}-A^T(w^n-w^{n+1})\right)\rangle
\end{displaymath}
where we have also used $\bar u^{n+1}=u^{n+1}-A^T(w^n-w^{n+1})$.
Similarly one finds from the first line of (\ref{fixedpoint}) and application of lemma~\ref{lemma} (with $u^+=\hat w$, $u^-=\hat w$, $\Delta= A\hat u$, $u=w^{n+1}$ and $P_C=P_\mu$) that:
\begin{displaymath}
\|\hat w-w^{n+1}\|^2\leq\|\hat w-w^{n+1}\|^2-\|\hat w-\hat w\|^2
 -2\langle w^{n+1}- \hat w, A\hat u\rangle.
\end{displaymath}
Together these two inequalities yield:
\begin{equation}
\begin{array}{lcl}
\|w^{n+1}-\hat w\|^2&\leq&\|w^{n}-\hat w\|^2-\|w^{n+1}-w^n\|^2\\
&&\qquad  -2\langle \hat w-w^{n+1},A\left((u^{n+1}-\hat u)-A^T(w^n-w^{n+1})\right)\rangle\\
&=&\|w^{n}-\hat w\|^2-\|w^{n+1}-w^n\|^2+\|A^T(w^{n+1}-\hat w)\|^2\\
&&\qquad +\|A^T(w^{n+1}-w^n)\|^2 -\|A^T(w^{n}-\hat w)\|^2\\
&&\qquad  -2\langle \hat w-w^{n+1},A(u^{n+1}-\hat u)\rangle.
\end{array}\label{tempw}
\end{equation}
In the same way, the third line of algorithm (\ref{combinedfull}) and lemma~\ref{lemma} (with $u^+=u^{n+1}$, $u^-=u^n$, $\Delta=-K^T(v^n+Ku^n-z^n)-A^T w^{n+1}$, $u=\hat u$, $P_C=\mathrm{Id}$) implies
\begin{displaymath}
\|u^{n+1}-\hat u\|^2\leq\|u^{n}-\hat u\|^2-\|u^{n+1}-u^n\|^2-2\langle \hat u-u^{n+1}, -K^T(v^n+Ku^n-z^n)-A^T w^{n+1}\rangle
\end{displaymath}
and the second line of (\ref{fixedpoint}) with lemma~\ref{lemma} (with $u^+=\hat u$, $u^-=\hat u$, $\Delta=-K^T(\hat v+K\hat u-\hat z)-A^T \hat w$, $u=u^{n+1}$, $P_C=\mathrm{Id}$) implies:
\begin{displaymath}
\|\hat u-u^{n+1}\|^2\leq\|\hat u-u^{n+1}\|^2-\|\hat u-\hat u\|^2
 -2\langle u^{n+1}- \hat u, -K^T(\hat v+K\hat u-\hat z)-A^T \hat w\rangle,
\end{displaymath}
such that together they yield:
\begin{equation}
\begin{array}{lcl}
\|u^{n+1}-\hat u\|^2&\leq&\|u^{n}-\hat u\|^2-\|u^{n+1}-u^n\|^2\\
&&-2\langle \hat u-u^{n+1}, -K^T\left(v^n-\hat v+K(u^n-\hat u)-(z^n-\hat z)\right)-A^T (w^{n+1}-\hat w)\rangle\\
&\leq&\|u^{n}-\hat u\|^2-\|u^{n+1}-u^n\|^2\\
&&\qquad+2\langle K(\hat u-u^{n+1}),K(u^n-\hat u)\rangle-2\langle K(\hat u-u^{n+1}),z^n-\hat z\rangle\\
&&\qquad+2\langle \hat u-u^{n+1}, K^T(v^n-\hat v)+A^T (w^{n+1}-\hat w)\rangle\\
&=&\|u^{n}-\hat u\|^2 -\|u^{n+1}-u^n\|^2 \\
&& \qquad-\|K(u^{n+1}-\hat u)\|^2+\|K(u^{n+1}-u^n)\|^2-\|K(u^{n}-\hat u)\|^2\\
&& \qquad+\|K(u^{n+1}-\hat u)\|^2-\|\hat z-z^n-K(\hat u-u^{n+1})\|^2+\|z^n-\hat z\|^2\\
&& \qquad-2\langle \hat u-u^{n+1},-K^T\left(v^n-\hat v\right)-A^T (w^{n+1}-\hat w)\rangle\\
&=&\|u^{n}-\hat u\|^2 -\|u^{n+1}-u^n\|^2+\|K(u^{n+1}-u^n)\|^2 \\
&& \qquad-\|K(u^{n}-\hat u)\|^2-\|\hat z-z^n-K(\hat u-u^{n+1})\|^2+\|z^n-\hat z\|^2\\
&& \qquad-2\langle \hat u-u^{n+1},-K^T\left(v^n-\hat v\right)-A^T (w^{n+1}-\hat w)\rangle.
\end{array}\label{tempu}
\end{equation}
And from the fourth line of algorithm (\ref{combinedfull}) and lemma~\ref{lemma} (with $u^+=z^{n+1}$, $u^-=0$, $\Delta=Ku^{n+1}+v^n$, $u=\hat z$ and $P_C=Q_{y,\epsilon}$) one finds:
\begin{displaymath}
\|z^{n+1}-\hat z\|^2\leq\|0-\hat z\|^2-\|z^{n+1}-0\|^2
-2\langle \hat z-z^{n+1},Ku^{n+1}+v^n\rangle.
\end{displaymath}
Similarly, from the third line of (\ref{fixedpoint}) and  lemma~\ref{lemma} with $u^+=\hat z$, $u^-=0$, $\Delta=K\hat u+\hat v$, $u=z^{n+1}$ and $P_C=Q_{y,\epsilon}$ one has:
\begin{displaymath}
\|z^{n+1}-\hat z\|^2\leq\|0-z^{n+1}\|^2-\|\hat z-0\|^2
-2\langle z^{n+1}-\hat z,K\hat u+\hat v\rangle
\end{displaymath}
which together yield:
\begin{displaymath}
2\|z^{n+1}-\hat z\|^2\leq 2\langle z^{n+1}-\hat z,K(u^{n+1}-\hat u)+v^n-\hat v\rangle
\end{displaymath}
or:
\begin{equation}
\begin{array}{lcl}
0&\leq& 2\langle z^{n+1}-\hat z,-(z^{n+1}-\hat z)+K(u^{n+1}-\hat u)+v^n-\hat v\rangle\\
&=& -\|z^{n+1}-\hat z\|^2-\|\hat z-z^{n+1}+K(u^{n+1}-\hat u)\|^2+\|K(u^{n+1}-\hat u)\|^2\\
&& +2\langle z^{n+1}-\hat z,v^n-\hat v\rangle.\\
&=& -\|z^{n+1}-\hat z\|^2-\theta^{-2}\|v^{n+1}-v^n\|^2+\|K(u^{n+1}-\hat u)\|^2\\
&& +2\langle z^{n+1}-\hat z,v^n-\hat v\rangle,\\ \label{tempz}
\end{array}
\end{equation}
where we have used the last line of (\ref{combinedfull}) and of (\ref{fixedpoint}).

Finally from the last line of (\ref{combinedfull}) and from lemma~\ref{lemma} (with $u^+=v^{n+1}$, $u^-=v^n$, $\Delta=\theta(Ku^{n+1}-z^{n+1})$, $u=\hat v$ and $P_C=\mathrm{Id}$) it follows that:
\begin{displaymath}
\|v^{n+1}-\hat v\|^2\leq\|v^{n}-\hat v\|^2-\|v^{n+1}-v^n\|^2-2\theta\langle \hat v- v^{n+1}, Ku^{n+1}-z^{n+1} \rangle.
\end{displaymath}
In the same way it follows from the last line of the fixed-point equation (\ref{fixedpoint}) and lemma~\ref{lemma} (with $u^+=\hat v$, $u^-=\hat v$, $\Delta=\theta(K\hat u-\hat z)$, $u=v^{n+1}$ and $P_C=Id$) that:
\begin{displaymath}
\|v^{n+1}-\hat v\|^2\leq\|v^{n+1}-\hat v\|^2-\|\hat v-\hat v\|^2-2\theta\langle v^{n+1}-\hat v, K\hat u-\hat z \rangle
\end{displaymath}
and together the last two expressions yield:
\begin{equation}
\begin{array}{lcl}
\theta^{-1}\|v^{n+1}-\hat v\|^2&=& \theta^{-1}\|v^{n}-\hat v\|^2 -\theta^{-1}\|v^{n+1}-v^n\|^2\\
&&\qquad+2\langle v^{n+1}-\hat v, K(u^{n+1}-\hat u)-(z^{n+1}-\hat z) \rangle\\
&=& \theta^{-1}\|v^{n}-\hat v\|^2 -\theta^{-1}\|v^{n+1}-v^n\|^2\\
&&\qquad +2\langle v^{n+1}- v^n, K(u^{n+1}-\hat u)-(z^{n+1}-\hat z) \rangle\\
&&\qquad+2\langle v^{n}- \hat v, K(u^{n+1}-\hat u)-(z^{n+1}-\hat z) \rangle\\
&\stackrel{\mathrm{last\ lines\ of\ (\ref{combinedfull})\ and\ (\ref{fixedpoint})}}{=}& \theta^{-1}\|v^{n}-\hat v\|^2 +\theta^{-1}\|v^{n+1}-v^n\|^2\\
&&\qquad+2\langle v^{n}- \hat v, K(u^{n+1}-\hat u)-(z^{n+1}-\hat z) \rangle.
\end{array}\label{tempv}
\end{equation}

Adding inequalities (\ref{tempw}), (\ref{tempu}), (\ref{tempz}) and (\ref{tempv}) together yields:
\begin{equation}
\begin{array}{l}
\|u^{n+1}-\hat u\|^2+\|w^{n+1}-\hat w\|^2+\theta^{-1}\|v^{n+1}-\hat v\|^2 \leq
\|u^{n}-\hat u\|^2 -\|u^{n+1}-u^n\|^2 \\
\qquad +\|K(u^{n+1}-\hat u)\|^2  +\|K(u^{n+1}-u^n)\|^2-\|K(u^{n}-\hat u)\|^2\\
\qquad -\|\hat z-z^n-K(\hat u-u^{n+1})\|^2+\|z^n-\hat z\|^2\\
\qquad+\|w^{n}-\hat w\|^2-\|w^{n+1}-w^n\|^2+\|A^T(w^{n+1}-\hat w)\|^2\\
\qquad+\|A^T(w^{n+1}-w^n)\|^2 -\|A^T(w^{n}-\hat w)\|^2\\
\qquad -\|z^{n+1}-\hat z\|^2-\theta^{-2}\|v^{n+1}-v^n\|^2+\theta^{-1}\|v^{n}-\hat v\|^2+\theta^{-1}\|v^{n+1}-v^n\|^2
\end{array}
\end{equation}
as all remaining inner products cancel.
As $\|K\|<1$ and $\|A\|<1$ we can introduce real regular square matrices $L$ and $B$ such that $L^TL=1-K^T K$ and $B^TB=1-AA^T$. The last inequality becomes:
\begin{equation}
\begin{array}{l}
\|L(u^{n+1}-\hat u)\|^2+\|B(w^{n+1}-\hat w)\|^2+\theta^{-1}\|v^{n+1}-\hat v\|^2+\|z^{n+1}-\hat z\|^2 \\
\qquad\quad\leq\|L(u^{n}-\hat u)\|^2+\|B(w^{n}-\hat w)\|^2+\theta^{-1}\|v^{n}-\hat v\|^2+\|z^n-\hat z\|^2 \\
\qquad\qquad-\|L(u^{n+1}-u^n)\|^2-\|B(w^{n+1}-w^n)\|^2 -\|\hat z-z^n-K(\hat u-u^{n+1})\|^2
\end{array}\label{temp1}
\end{equation}
where we have used that $-\theta^{-2}\|v^{n+1}-v^n\|^2+\theta^{-1}\|v^{n+1}-v^n\|^2\leq0$ (for $0<\theta\leq1$).

It follows from inequality (\ref{temp1}) that $(u^n,w^n,v^n,z^n)_n$ is a bounded sequence and that a convergent subsequence $(u^{n_j},w^{n_j},v^{n_j},z^{n_j})_j$ exists. We call the limit of this subsequence $(u^\dagger,w^\dagger,v^\dagger,z^\dagger)$.

Summing inequality (\ref{temp1}) from $M$ to $N-1$ yields:
\begin{equation}
\begin{array}{l}
\|L(u^{N}-\hat u)\|^2+\|B(w^{N}-\hat w)\|^2+\theta^{-1}\|v^{N}-\hat v\|^2+\|z^{N}-\hat z\|^2\\
\quad\leq\|L(u^{M}-\hat u)\|^2+\|B(w^{M}-\hat w)\|^2+\theta^{-1}\|v^{M}-\hat v\|^2+\|z^M-\hat z\|^2 \\ \qquad-\sum_{n=M}^{N-1}\left(\|L(u^{n+1}-u^n)\|^2+\|B(w^{n+1}-w^n)\|^2 +\|\hat z-z^n-K(\hat u-u^{n+1})\|^2\right).
\end{array}\label{temp2}
\end{equation}
It follows that the sum in the right hand side of this expression bounded (independent of $N$) and therefore that $\|L(u^{n+1}-u^n)\|$, $\|B(w^{n+1}-w^n)\|$ and $\|\hat z-z^n-K(\hat u-u^{n+1})\|$ tend to zero when $n$ tends to infinity. As $B$ and $L$ are regular, this implies that $\|u^{n+1}-u^n\|$ and $\|w^{n+1}-w^n\|$ tend to zero for large $n$. Then also $\|\hat z-z^n-K(\hat u-u^{n})\|$ tends to zero for large $n$. The iteration (\ref{combinedfull}) and relations (\ref{fixedpoint}) imply that $v^{n+1}=v^n+\theta (K(u^{n+1}-\hat u)-(z^{n+1}-\hat z))$, and we therefore find that $\|v^{n+1}-v^n\|\stackrel{n\rightarrow\infty}{\longrightarrow}0$. It follows also from the iteration (\ref{combinedfull}) and the previous remarks that:
\begin{displaymath}
\begin{array}{lcl}
\|z^{n+1}-z^n\|&=&\|Q_{y,\epsilon}(Ku^{n+1}+v^n)-Q_{y,\epsilon}(Ku^{n}+v^{n-1})\|\\
&\leq& \|(Ku^{n+1}+v^n)-(Ku^{n}+v^{n-1})\|\\
&\leq& \|K(u^{n+1}-u^n)\|+\|v^n-v^{n-1}\|\stackrel{n\rightarrow\infty}{\longrightarrow}0.
\end{array}
\end{displaymath}
We can therefore conclude that $(u^{n_j+1},w^{n_j+1},v^{n_j+1},z^{n_j+1}) \stackrel{j\rightarrow\infty}{\longrightarrow}(u^\dagger,w^\dagger,v^\dagger,z^\dagger)$ as well. This implies then that $(u^\dagger,w^\dagger,v^\dagger,z^\dagger)$ is a fixed point of the iteration (\ref{combinedfull}) and satisfies the equations (\ref{fixedpoint}).

We now choose $(\hat u,\hat w,\hat v,\hat z)=(u^\dagger,w^\dagger,v^\dagger,z^\dagger)$ in inequality (\ref{temp2}) and find that:
\begin{displaymath}
\begin{array}{l}
\|L(u^{N}- u^\dagger)\|^2+\|B(w^{N}- w^\dagger)\|^2+ \theta^{-1}\|v^{N}- v^\dagger\|^2+\|z^{N}-z^\dagger\|^2\\
\qquad\leq\|L(u^{M}-u^\dagger)\|^2+\|B(w^{M}-w^\dagger)\|^2+\theta^{-1}\|v^{M}-v^\dagger\|^2 +\|z^M-z^\dagger\|^2
\end{array}
\end{displaymath}
for $N>M$.
As there is a convergent subsequence, the right hand side can be made as small as desired by choice of $M$, and it follows that the entire sequence $(u^{n},w^{n},v^{n},z^{n})$ tends to the fixed point $(u^\dagger,w^\dagger,v^\dagger,z^\dagger)$ for large $n$.
\end{proof}

One can adapt the algorithm (\ref{combialg}) (or (\ref{combinedfull})) and the proof to suit the problem $\min_x H(Ax)+J(Kx)$ (with $H$ and $J$ two convex functions). In this case the projection operator $P_\mu$ must be replaced by the proximity operator $\mathrm{prox}_{H^\ast}$ and the projection operator $Q_{y,\epsilon}$ must be replaced by the proximity operator of $J$: $\mathrm{prox}_{J}$. As long as these (non-linear) operators have simple expressions, the above algorithms are useable in those cases as well. See \citep{Combettes.Pesquet2011} for a general review of proximity operators and \citep{Loris.Verhoeven2011} on how proximity operators were used for a generalization of  problem (\ref{penproblem}) and algorithm (\ref{gista}).

\section{Conclusions}

A set of simple iterative algorithms for the minimization of a penalized least squares functional (\ref{penproblem}) was presented and their convergence speed was compared numerically in a seismic tomography context ($K\neq \mathrm{Id}$). For the examples considered, the comparison shows that all of these algorithm can produce the minimizer of the functional to within $1$ to $10\%$ after about $1000$ iterations. The `best' algorithm depends on the step sizes, the penalty parameter $\lambda$ and on the metric chosen (functional or distance to minimizer). They can therefore all be used in practice for solving these kind of problems as they appear in seismic tomography.

The four algorithms for minimizing the penalized functional (\ref{penproblem}) are not new. Here they  are presented in a uniform notation making comparison and implementation easier. We also presented and compared two iterative algorithms for the (equivalent) constrained formulation (\ref{conproblem}) of this problem. One of these algorithms is new and a proof of convergence is included.

The advantage of the constrained formulation (\ref{conproblem}) over the penalized formulation (\ref{penproblem}) is that often the penalty parameter $\lambda$ has to be determined by trial and error so as to fit the data to the desired level $\|K\hat u(\lambda)-y\|\leq\epsilon$. This is done automatically in formulation (\ref{conproblem}). We have listed simple iterative algorithms for both approaches. However, if the operator $A$ is invertible (and the inverse is readily available), then it can be more advantageous to make a change of variables in the \emph{penalized} formulation and to use an accelerated algorithm (see discussion at end of Section~\ref{penalgsection}).

For the readers' convenience we wrote down the explicit forms of several non-smooth convex penalties that can be used for edge-preserving regularization of a seismic imaging problem (total variation penalties and various generalizations) and that can be solved with these algorithms. We also wrote explicit formulas for the convex projection operators that are used by the iterative algorithms in these cases.

We discussed the qualitative properties of these reconstructions on a synthetic seismic tomography example. Particular emphasis was given to the role of sparse local differences. For the case of image denoising $K=\mathrm{Id}$ with TV-like penalties one may refer to \citep{Strong2003} and \citep{Setzer2011}. Still in the case of denoising, other penalties on local differences (other than $\ell_1$-norm type, including non-convex ones), with the goal of maintaining edges in the reconstruction, are compared numerically in \citep{Lukic2011}.

\section*{Acknowledgements}
The authors would like to acknowledge fruitful discussions with F. Simons, K. Sigloch and G. Nolet.
I.L. is a research associate of the Fonds de la recherche Scientifique-FNRS.
Part of this research was started while the authors were at CAMP department of the Vrije Universiteit Brussel and was supported by Vrije Universiteit Brussel grant GOA-062 and by the Fonds voor Wetenschappelijk Onderzoek-Vlaanderen grant G.0564.09N.

\bibliography{geomathTV}{}
\bibliographystyle{abbrvnat}

\end{document}